\documentclass[final,5p,times,twocolumn]{elsarticle}
\usepackage{amsmath,xcolor}
\usepackage{graphicx}
\usepackage{epstopdf}
\usepackage{url}
\usepackage{amsthm}
\usepackage{comment}
\usepackage{amssymb}
\usepackage{caption}
\usepackage{subcaption}
\usepackage{amsfonts}
\usepackage[ruled,norelsize,linesnumbered]{algorithm2e}
\usepackage{soul,xcolor}
\usepackage{bm}
\usepackage{tabularx}
\usepackage{float}
\usepackage{array}
\makeatletter
\newcommand{\removelatexerror}{\let\@latex@error\@gobble}
\makeatother
\biboptions{sort&compress}
\journal{Automatica}
\newtheorem{theorem}{Theorem}

\newtheorem{remark}{Remark}
\newtheorem{definition}{Definition}
\usepackage{slashbox}

\makeatletter
\def\hlinewd#1{%
\noalign{\ifnum0=`}\fi\hrule \@height #1 %
\futurelet\reserved@a\@xhline}
\makeatother
\removelatexerror
\newcolumntype{L}[1]{>{\raggedright\let\newline\\\arraybackslash\hspace{0pt}}m{#1}}
\newcolumntype{C}[1]{>{\centering\let\newline\\\arraybackslash\hspace{0pt}}m{#1}}
\newcolumntype{R}[1]{>{\raggedleft\let\newline\\\arraybackslash\hspace{0pt}}m{#1}}
\begin{document}

\begin{frontmatter}

\title{Privacy-Preserving Dynamic Average Consensus via State Decomposition: Case Study on Multi-Robot Formation Control}

\author[MSU]{Kaixiang Zhang}\ead{zhangk64@msu.edu}
\author[MSU]{Zhaojian Li\corref{cor1}}\ead{lizhaoj1@egr.msu.edu}
\cortext[cor1]{Corresponding author Z. Li. Tel. +1-5174321821.}
\author[Clemson]{Yongqiang Wang}\ead{yongqiw@clemson.edu}
\author[Saudi]{Ali Louati}\ead{a.louati@psau.edu.sa}
\author[ZJU]{Jian Chen}\ead{jchen@zju.edu.cn}             % e-mail address

\address[MSU]{Department of Mechanical Engineering, Michigan State University, East Lansing, MI 48824, USA.}
\address[Clemson]{Department of Electrical and Computer Engineering, Clemson University, Clemson, SC 29634, USA.}
\address[Saudi]{Information Systems Department, Prince Sattam Bin Abdulaziz University, 11942 Alkharj, Kingdom of Saudi Arabia.}
\address[ZJU]{State Key Laboratory of Industrial Control Technology, School of Mechanical Engineering, Zhejiang University, Hangzhou 310027, China.\vspace{-20pt}}

\begin{abstract}
Dynamic average consensus is a decentralized control/estimation framework where a group of agents cooperatively track the average of local time-varying reference signals. In this paper, we develop a novel state decomposition-based privacy preservation scheme to protect the privacy of agents when sharing information with neighboring agents. Specifically, we first show that an external eavesdropper can successfully wiretap the reference signals of all agents in a conventional dynamic average consensus algorithm. To protect privacy against the eavesdropper, a state decomposition scheme is developed where the original state of each agent is decomposed into two sub-states: one succeeds the role of the original state in inter-node interactions, while the other sub-state only communicates with the first one and is invisible to other neighboring agents. Rigorous analyses are performed to show that 1) the proposed privacy scheme preserves the convergence of the average consensus; and 2) the privacy of the agents is protected such that an eavesdropper cannot discover the private reference signals with any guaranteed accuracy. The developed privacy-preserving dynamic average consensus framework is then applied to the formation control of multiple non-holonomic mobile robots, in which the efficacy of the scheme is demonstrated. 
Numerical simulation is provided to illustrate the effectiveness of the proposed approach.
\end{abstract}
\end{frontmatter}

\thispagestyle{empty}
\pagestyle{empty}
\setlength{\abovecaptionskip}{0pt}
\setlength{\belowcaptionskip}{4pt}
\setlength{\textfloatsep}{0pt}

\section{Introduction} 

Average consensus has been extensively studied in recent years. It underpins many advantages of distributed systems and is emerging as an effective tool for diverse applications, including sensor fusion \cite{Olfati-SaberCDC2005,AraguesTRO2012}, distributed resource allocation \cite{kia2017SCL}, and multi-agent coordination \cite{Porfiri2007AUTO,MontijanoTRO2016}. 
%However, in some scenarios involving sensitive information, the adoption of average consensus algorithms is impeded by concerns over the privacy preservation guarantees. Motivated by the demand for protecting sensitive information, this paper considers the privacy preservation issue in the dynamic average consensus using a state decomposition scheme.
%Existing average consensus algorithms can be roughly categorized into static cases \cite{Olfati-Saber2007,Ren2008distributed,Ren2010distributed} and dynamic cases \cite{Zhu2010AUTO,Spanos2005IFAC,Freeman2006CDC,Bai2010CDC,ChenTAC2012,ChenTAC2015}. Specifically, based on local interacting information, the agents in static cases aim at reaching agreement on the average of initial agent states, while the dynamic average consensus is to design a distributed update law such that all the agents can track the average of locally available time-varying reference signals. 
Based on the type of signals to be averaged, average consensus algorithms can be categorized as static \cite{Olfati-Saber2007,Ren2008distributed,Ren2010distributed} or dynamic \cite{Zhu2010AUTO,Spanos2005IFAC,Freeman2006CDC,Bai2010CDC,ChenTAC2012,ChenTAC2015}, where in static average consensus agents seek to reach agreement on the average of initial agent states, whereas the dynamic average consensus is to design a distributed update law such that all agents can track the average of locally available time-varying reference signals. As dynamic average consensus has many emerging applications in distributed control and estimation, it will be the focus of this paper.

So far, different approaches have been proposed to address the dynamic average consensus problems. The initial work \cite{Spanos2005IFAC} designs a consensus algorithm that can track the average of reference signals with steady states. Based on input-to-state stability property, a proportional-integral (PI) algorithm is proposed in \cite{Freeman2006CDC} to achieve consensus task for slowly time-varying and static reference signals. The PI algorithm is further generalized in \cite{Bai2010CDC} and can converge with zero steady-state error if the Laplace transform of reference signals has a common monic denominator polynomial. Moreover, nonsmooth algorithms are developed in \cite{ChenTAC2012,ChenTAC2015} to accomplish finite-time consensus convergence.
%Different from the static average consensus that aims at reaching agreement on the average of static reference signals, the dynamic average consensus is to design an update law for each agent based on local interacting information such that all the agents can track the average of locally available time-varying reference signals.

In the aforementioned dynamic average consensus methods, each agent needs to share the explicit state value with its neighboring agents, which can breach the privacy of participating agents as the state values typically contain privacy-sensitive information such as its local reference signals. 
%Disclosing the state information via this way induces potential privacy problems. The state values contain private information regarding local reference signals, and thus the disclosure of state information might breach the privacy of participating agents. 
For example, an eavesdropper can wiretap the communication channel and infer privacy-sensitive information based on the exchanged signals. 
%Actually, if there exist eavesdroppers who want to steal information by wiretapping communication channels, exchanging agents' explicit state values through the communication network is vulnerable.
Considering the aforementioned issue and growing awareness of security, it is an urgent need to protect the agents' privacy in dynamic average consensus. 
%By now, several privacy-preserving strategies have been investigated for the static average consensus. One interesting idea is to inject carefully-designed perturbation signals into the state values of the agents, which masks private information and can obfuscate eavesdroppers \cite{MoTAC2017,Nozari2017AUTO,rezazadeh2019arxiv}. 
So far the privacy preservation schemes have been mainly focused on static average consensus while results on the privacy of dynamic average consensus are very sparse. For examples, studies in \cite{MoTAC2017,Nozari2017AUTO,rezazadeh2019arxiv} use carefully designed perturbation signals to obfuscate the states to protect the transmitted state information.
Another idea is to exploit cryptography to improve the resilience of static average consensus algorithms to privacy attacks \cite{Freris2016,RuanTAC2019,HadjicostisTAC2020}. In our prior work \cite{WangTAC2019}, a state decomposition method is developed for static average consensus, which can guarantee the privacy of all participating agents and is light-weight in calculation. 
%Note that most of the existing privacy-preserving strategies \cite{MoTAC2017,Nozari2017AUTO,rezazadeh2019arxiv,Freris2016,RuanTAC2019,HadjicostisTAC2020,WangTAC2019} are proposed for static average consensus, and few results focus on the privacy issue of dynamic average consensus. The static and dynamic average consensus algorithms have diverse agreement objectives and exploit different protocols to update agents' states. Therefore, the privacy preservation schemes for static average consensus cannot be directly applicable to the dynamic one. 
In this work, we extend the state decomposition technique to dynamic average consensus. Note that as dynamic average consensus involves time-varying reference signals, the privacy design in static average consensus cannot be directly applied \cite{Nozari2017AUTO,KiaCSM2019}. 
%The privacy design and analysis therein is more challenging as compared to those in static average consensus.  
The privacy design and analysis herein rely upon algebraic graph theory and Lyapunov-based control theory, which are more challenging as compared to those in \cite{WangTAC2019}. 

In this paper, we first illustrate the necessity of privacy protection by designing an attack model where an external eavesdropper can wiretap the reference signals when the agent states are updated following a conventional dynamic average consensus algorithm. We will then develop a state decomposition scheme to decompose the original state of each agent into two sub-states to achieve privacy preservation in dynamic average consensus. Specifically, one sub-state takes the role of the original state to interact with neighboring agents, while the other sub-state is invisible to the outside world and only exchanges information with the first sub-state. To ensure that the consensus algorithm with state decomposition retains the same average consensus results as the conventional method, the reference signals of the two sub-states are selected with their mean equal to the reference signals of the original state. We rigorously show that the proposed state decomposition scheme can protect the private reference signals from being inferred by the external eavesdropper. In addition, a case study on formation control of non-holonomic mobile robots is presented, which shows that the developed approach can be integrated with the tracking controller to accomplish privacy-preserving distributed control. Simulation results are given to validate the performance of the state decomposition scheme.

The remainder of this paper is organized as follows. Section \uppercase\expandafter{2} introduces a conventional dynamic average consensus algorithm and the privacy definition. A corresponding eavesdropping strategy is then presented in Section \uppercase\expandafter{3}. The state decomposition scheme for privacy preservation is developed in Section \uppercase\expandafter{4} and further exploited in Section \uppercase\expandafter{5} to accomplish the formation control of non-holonomic mobile robots.    
Simulation results are provided in Section \uppercase\expandafter{6}. Finally, the conclusion is given in Section \uppercase\expandafter{7}.

%The contributions of this paper include the following. Firstly, we propose a novel and practical resource allocation model for cloud-aided automotive systems that combines a high-level auction-based resource bidding with a low-level onboard optimization.  Secondly, we characterize the auction-based interplay as a multi-player game and establish the existence and uniqueness of the Nash equilibrium. Furthermore, two decentralized schemes that are secure and require little communication are proposed. Results on the convergence to the unique Nash equilibrium are derived. Last but not least, numerical simulations are presented to demonstrate the efficacy of the developed framework. In particular, we characterize the scheme's robustness to stochastic environment using total variance distance and statistically show their stochastic stability under random task arrival rate.

%The rest of the paper is organized as follows. In Section II, the resource allocation problem for cloud-aided automotive systems is described and a hierarchical allocation model is developed. In Section III, results on the existence and uniqueness of Nash equilibrium of the induced game are established. Dynamics of reaching the equilibrium is presented in Section IV, where synchronized and asynchronized update schemes are proposed with proved convergence to the unique Nash equilibrium. Numerical simulations are presented in Section V whereas Section VI concludes the paper.

\section{Preliminaries} \label{sec_preliminary}
\subsection{Dynamic Average Consensus}
We first review the dynamic average consensus problem.
%Suppose that there are $n$ time-varying reference signals $r_{i}(t) \in \mathbb{R}^{m}$, $i=1,2,\cdots, n$, satisfying the following dynamics $\dot{r}_{i}(t) = f_{i}(t)$. Suppose that there are $n$ agents with $x_{i}(t) \in \mathbb{R}^{m}$ being the state of agent $i$, $i=1,2,\cdots, n$. 
Consider $n$ agents where each has a time-varying reference signal $r_{i}(t) \in \mathbb{R}^{m}$, $i=1,2,\cdots, n$, satisfying the following dynamics: 
%$\dot{r}_{i}(t) = f_{i}(t)$.
\begin{equation} \label{dot_r}
\dot{r}_{i}(t) = f_{i}(t).
\end{equation}
In \eqref{dot_r} $f_{i}(t) \in \mathbb{R}^{m}$ is the derivative of the reference signal.
Each agent $i$ has access to $r_{i}(t)$, $f_{i}(t)$ as well as information from a subset of the other agents. This subset is referred to as the neighborhood of agent $i$ and denoted by $\mathcal{N}_{i}$. A graph $\mathcal{G} \triangleq \left\lbrace \mathcal{V}, \mathcal{E} \right\rbrace $ is used to describe the network topology between the agents, where $\mathcal{V} \triangleq \left\lbrace 1, 2, \cdots, n \right\rbrace$ is the node set and $\mathcal{E} \triangleq \left\lbrace (i, j) | i\in \mathcal{N}_{j}, j=1, 2, \cdots, n \right\rbrace$ is the edge set. We consider undirected graph where $j \in \mathcal{N}_{i}$ implies $i \in \mathcal{N}_{j}$. A graph is connected if and only if there is a path from any node to any other node. The adjacency matrix $A \triangleq \left[ a_{ij} \right] \in \mathbb{R}^{n\times n}$ of $\mathcal{G}$ is defined as follows: $a_{ij} =1$ if $(i, j)\in \mathcal{E}$; $a_{ij} =0$ otherwise. The degree of agent $i$ is $d_{i} \triangleq \sum_{j=1}^{n}a_{ij}$, and the degree matrix is given by $D \triangleq \text{diag}\left(d_{1}, \cdots, d_{n}\right) \in \mathbb{R}^{n\times n}$. The Laplacian matrix of $\mathcal{G}$ is then defined as $L \triangleq D-A\in \mathbb{R}^{n\times n}$.
%\begin{assumption}
%	The graph is undirected and connected.
%\end{assumption}
%\begin{assumption}
%	The signals $r_{i}(t)$, $f_{i}(t)$, and $\dot{f}_{i}(t)$ are bounded.
%\end{assumption}

In addition, each agent $i$ keeps an internal state $x_{i}(t) \in \mathbb{R}^{m}$, $i=1,2,\cdots, n$, and the objective of the dynamic average consensus is to design a distributed algorithm, such that all agents will finally track the average of the $n$ time-varying reference signals, i.e., $\| x_{i}(t)-\frac{1}{n}\sum_{j=1}^{n}r_{j}(t) \| \rightarrow 0 $ as $t \rightarrow \infty$. 
Assuming that the graph is undirected and connected, and the signals $r_{i}(t)$ and $f_{i}(t)$ are bounded, the following algorithm can be exploited to achieve the dynamic average consensus \cite{Spanos2005IFAC}:
%To achieve the dynamic average consensus, one of the commonly used update laws is given by \cite{Spanos2005IFAC}
\begin{equation} \label{dot_x}
\begin{aligned}
\dot{x}_{i}(t) &= f_{i}(t) + \kappa \sum_{j=1}^{n}a_{ij}\left( x_{j}(t) - x_{i}(t) \right),
\\
x_{i}(0) &= r_{i}(0), \quad \forall i\in \mathcal{V},
\end{aligned}
\end{equation}
where $\kappa \in \mathbb{R}$ is a positive constant. 
%Using a Laplace domain analysis, \cite{Spanos2005IFAC} shows that, if each input signal $r_{i}$, $i=1, 2, \cdots,n$, has a Laplace transform with all poles in the left half-plane and at most one zero pole (such signals are asymptotically constant), all the agents implementing algorithm \eqref{dot_x} over a connected graph track $\frac{1}{n}\sum_{j=1}^{n}r_{j}(t)$ with zero error asymptotically. As shown in \cite{KiaCSM2019}, the convergence properties of algorithm \eqref{dot_x} can be described more comprehensively using time-domain input-to-state stability analysis.
Using a time-domain analysis, \cite{KiaCSM2019} shows that if each input signal $f_{i}(t)$, $i=1, 2, \cdots,n$, is bounded, all the agents implementing algorithm \eqref{dot_x} over a connected graph are input-to-state stable and the tracking errors are ultimately bounded. Moreover, the convergence rate to the error bound is no worse than $\kappa \lambda_{2}(L)$, where $\lambda_{2}(L) \in \mathbb{R}$ is the second smallest eigenvalue of the Laplacian matrix $L$. 

\subsection{Privacy Definition}
As can be seen from \eqref{dot_x}, the conventional dynamic average consensus algorithm involves the exchange of states among neighboring agents, which can leak privacy-sensitive information such as the local reference signals. In this paper, we consider an eavesdropping attacker who knows the network topology and can wiretap communication channels and access exchanged information. Specifically, we consider the case where the eavesdropper is interested in obtaining the reference signals $r_i(t)$ and $f_i(t)$. 

Based on the above discussion, the definition of privacy in dynamic average consensus is given by:
\begin{definition}
	For a connected graph with $n$ agents, the privacy of the reference signals $r_i(t)$ and $f_i(t)$ from agent $i$ is preserved if an external eavesdropper cannot estimate the values of $r_i(t)$ and $f_i(t)$ with any accuracy.
\end{definition}

This privacy definition requires that an eavesdropping adversary cannot even approximately estimate (e.g., find a finite range for) the private signals and thus is more stringent than the privacy definition in \cite{Liu2006TKDE,Han2010TKDE} which defines privacy preservation as the inability of an adversary to {\it uniquely} determine the protected value. Furthermore, the privacy-preserving dynamic average consensus has potential applications in emerging distributed automated systems. For example, multiple power generators in the smart grid can exploit the consensus scheme to reach agreement on the average planned power generation of the whole network and keep their individual evolving generation plan secret since the generation plan is sensitive in bidding the right for energy selling \cite{Fang2012}. Another example is the formation control of multiple mobile robots, which will be detailed in Section \ref{sec_formation}. 

\section{Privacy Attack Model} \label{sec_eavesdrop}

%each agent updates its state by using local reference signals, and hence exchanging state values via the communication network has a high risk to reveal the private information regarding reference signals.  
%Suppose that an external eavesdropper is interested in obtaining the reference signals $r_{i}(t)$ and $f_{i}(t)$ of agent $i$. Specifically, the eavesdropper is an external attacker who knows the network topology and can wiretap communication channels and access exchanged information. 
We now show that the dynamic consensus algorithm \eqref{dot_x} is not privacy preserving, that is, the external eavesdropper can successfully obtain the reference signals $r_{i}(t)$ and $f_{i}(t)$ when agents follow the consensus algorithm in \eqref{dot_x}.

In particular, let $\hat{x}_{i}(t) \in \mathbb{R}^{m}$, $\hat{r}_{i}(t) \in \mathbb{R}^{m}$, and $\hat{f}_{i}(t) \in \mathbb{R}^{m}$ be the eavesdropper's estimates of $x_{i}(t)$, $r_{i}(t)$, and $f_{i}(t)$, respectively. An observer based attack model can be designed to estimate $r_{i}(t)$ and $f_{i}(t)$ as follows:
\begin{equation} \label{observer}
\begin{aligned}
\dot{\hat{x}}_{i}(t) &= \hat{f}_{i}(t) + \kappa \sum_{j=1}^{n}a_{ij}\left( x_{j}(t) - x_{i}(t) \right) + k_{1} \tilde{x}_{i}(t),
\\
\dot{\hat{r}}_{i}(t) &= k_{2}\left( x_{i}(t)-z_{i}(t)-\hat{r}_{i}(t) \right) + \hat{f}_{i}(t),
\\
\hat{f}_{i}(t) &= k_{3}x_{i}(t) + \hat{f}_{i}'(t),
\\
\dot{\hat{f}}_{i}'(t) &= -k_{3}\left( \hat{f}_{i}(t) + \kappa \sum_{j=1}^{n}a_{ij}\left( x_{j}(t) - x_{i}(t) \right) \right) + k_{4}\tilde{x}_{i}(t),
\end{aligned}
\end{equation}
where $k_{1}$, $k_{2}$, $k_{3}$, $k_{4} \in \mathbb{R}$ are positive constants to be designed, $\tilde{x}_{i}(t) \triangleq x_{i}(t)-\hat{x}_{i}(t) \in \mathbb{R}^{m}$ is the estimation error, $\hat{f}_{i}'(t) \in \mathbb{R}^{m}$ is an auxiliary variable, and $z_{i}(t) \in \mathbb{R}^{m}$ is the local filter updated by
\begin{equation} \label{dot_z}
\begin{aligned}
\dot{z}_{i}(t) &= \kappa \sum_{j=1}^{n}a_{ij}\left( x_{j}(t) - x_{i}(t) \right),
\\
z_{i}(0) &= 0. 
\end{aligned}
\end{equation}

\begin{theorem}
	When the algorithm in \eqref{dot_x} is utilized to achieve dynamic average consensus, the external eavesdropper can infer the reference signals $r_{i}(t)$ and $f_{i}(t)$ by using the observer in \eqref{observer}. More precisely, assuming that the signals $r_{i}(t)$, $f_{i}(t)$, and $\dot{f}_{i}(t)$ are bounded, i.e., $r_{i}(t)$, $f_{i}(t)$, $\dot{f}_{i}(t) \in \mathcal{L}_{\infty}$, then
	%$\left. 1\right)$ \eqref{observer} ensures that the estimation errors $\tilde{r}_{i}(t) \triangleq r_{i}(t)-\hat{r}_{i}(t)$, $\tilde{f}_{i}(t) \triangleq f_{i}(t)-\hat{f}_{i}(t) \in \mathbb{R}^{m}$ are uniformly ultimately bounded (UUB). $\left. 2\right)$ Moreover, if $\dot{f}_{i}(t)$ is in the $\mathcal{L}_{2}$-space, the estimation errors $\tilde{r}_{i}(t)$ and $\tilde{f}_{i}(t)$ converge to zero asymptotically.
	\begin{enumerate}
		\item The attacker using \eqref{observer} is guaranteed to obtain the private reference information in the sense that the estimation errors $\tilde{r}_{i}(t) \triangleq r_{i}(t)-\hat{r}_{i}(t)$, $\tilde{f}_{i}(t) \triangleq f_{i}(t)-\hat{f}_{i}(t) \in \mathbb{R}^{m}$ are uniformly ultimately bounded (UUB).
		\item If $\dot{f}_{i}(t)$ is in the $\mathcal{L}_{2}$-space, the estimation errors $\tilde{r}_{i}(t)$ and $\tilde{f}_{i}(t)$ converge to zero asymptotically.
	\end{enumerate}
\end{theorem}

\begin{proof}
	To prove the first claim, a non-negative Lyapunov function $V(t) \in \mathbb{R}$ is introduced, as follows:
	\begin{equation} \label{V}
	V(t) \triangleq \frac{1}{2}k_{4}\tilde{x}_{i}^{T}(t)\tilde{x}_{i}(t) + \frac{1}{2} \tilde{r}_{i}^{T}(t)\tilde{r}_{i}(t) + \frac{1}{2} \tilde{f}_{i}^{T}(t)\tilde{f}_{i}(t),
	\end{equation}
	from which it follows that $V(t)$ can be bounded by
	\begin{equation} \label{bound_V}
	\underline{\mu}y^{T}(t)y(t) \le V(t) \le \overline{\mu}y^{T}(t)y(t),
	\end{equation}
	where $\underline{\mu} \triangleq \min\left\lbrace \frac{1}{2}k_{4}, \frac{1}{2} \right\rbrace$, $\overline{\mu} \triangleq \max \left\lbrace \frac{1}{2}k_{4}, \frac{1}{2} \right\rbrace \in \mathbb{R}$, and $y(t) \triangleq \begin{bmatrix}
	\tilde{x}_{i}^{T}(t) & \tilde{r}_{i}^{T}(t) & \tilde{f}_{i}^{T}(t)
	\end{bmatrix}^{T} \in \mathbb{R}^{3m}$ is the augmented estimate vector.
	Taking the time derivative of \eqref{V} and substituting it in \eqref{dot_x}-\eqref{dot_z} yield
	\begin{equation} \label{dot_V}
	\begin{aligned}
	\dot{V}(t) &= k_{4}\tilde{x}_{i}^{T}(t) \dot{\tilde{x}}_{i}(t) + \tilde{r}_{i}^{T}(t)\dot{\tilde{r}}_{i}(t) + \tilde{f}_{i}^{T}(t)\dot{\tilde{f}}_{i}(t)
	\\
	&= -k_{1}k_{4}\tilde{x}_{i}^{T}(t)\tilde{x}_{i}(t) - k_{2}\tilde{r}_{i}^{T}(t)\tilde{r}_{i}(t) - k_{3}\tilde{f}_{i}^{T}(t)\tilde{f}_{i}(t) 
	\\
	& \quad \, + \tilde{r}_{i}^{T}(t)\tilde{f}_{i}(t) + \tilde{f}_{i}^{T}(t)\dot{f}_{i}(t)
	\\
	&\le -k_{1}k_{4}\tilde{x}_{i}^{T}(t)\tilde{x}_{i}(t) - \frac{k_{2}}{2}\tilde{r}_{i}^{T}(t)\tilde{r}_{i}(t)  
	\\
	& \quad \, - \left( \frac{k_{3}}{2}-\frac{1}{2k_{2}} \right)\tilde{f}_{i}^{T}(t)\tilde{f}_{i}(t) + \frac{1}{2k_{3}} \dot{f}_{i}^{T}(t)\dot{f}_{i}(t)
	\\
	&\le -\mu y^{T}(t)y(t) + \varrho(t),
	\end{aligned}
	\end{equation}
	where $\mu \triangleq \min \left\lbrace k_{1}k_{4}, \frac{k_{2}}{2}, \frac{k_{3}}{2}-\frac{1}{2k_{2}} \right\rbrace \in \mathbb{R}$ and $\varrho(t) \triangleq \frac{1}{2k_{3}} \dot{f}_{i}^{T}(t)\dot{f}_{i}(t) \in \mathbb{R}$. It is clear that $\mu$ is positive provided that $k_{2}$ and $k_{3}$ are chosen to satisfy $k_{3}>\frac{1}{k_{2}}$. Since $\dot{f}_{i}(t) \in \mathcal{L}_{\infty}$, $\varrho(t)$ is bounded.
	By utilizing \eqref{bound_V} and \eqref{dot_V}, Theorem 4.18 in \cite{khalil2002nonlinear} can be invoked to show that $y(t)$, i.e., $\tilde{x}_{i}(t)$, $\tilde{r}_{i}(t)$ and $\tilde{f}_{i}(t)$, is UUB.
	
	We now prove the second claim. Based on the assumption $\dot{f}_{i}(t) \in \mathcal{L}_{2}$, it can be obtained that there exists a bounded positive constant $\iota \in \mathbb{R}$ such that $\forall t\ge 0$,
	\[
	\int_{0}^{t} \frac{1}{2k_{3}} \dot{f}_{i}^{T}(\tau)\dot{f}_{i}(\tau)d\tau \le \iota.
	\]
	Let the non-negative function $W(t) \in \mathbb{R}$ be defined as
	\begin{equation} \label{W}
	W(t) \triangleq V(t) + \iota - \int_{0}^{t} \frac{1}{2k_{3}} \dot{f}_{i}^{T}(\tau)\dot{f}_{i}(\tau)d\tau.
	\end{equation}
	Taking the time derivative of \eqref{W} and utilizing \eqref{dot_V}, it can be concluded that
	\begin{equation} \label{dot_W}
	\begin{aligned}
	\dot{W}(t) &= -k_{1}k_{4}\tilde{x}_{i}^{T}(t)\tilde{x}_{i}(t) - \frac{k_{2}}{2}\tilde{r}_{i}^{T}(t)\tilde{r}_{i}(t) 
	\\
	& \quad \, 
	- \left( \frac{k_{3}}{2}-\frac{1}{2k_{2}} \right)\tilde{f}_{i}^{T}(t)\tilde{f}_{i}(t) \le 0.
	\end{aligned}
	\end{equation}
	According to \eqref{W} and \eqref{dot_W}, it follows that $W(t) \in \mathcal{L}_{\infty}$, i.e., $\tilde{x}_{i}(t)$, $\tilde{r}_{i}(t)$ $\tilde{f}_{i}(t) \in \mathcal{L}_{\infty} \bigcap \mathcal{L}_{2}$. The boundedness of $r_{i}(t)$, $f_{i}(t)$, $\dot{f}_{i}(t)$ and the expression in \eqref{observer} can be used to conclude that $\dot{\tilde{x}}_{i}(t)$, $\dot{\tilde{r}}_{i}(t)$, $\dot{\tilde{f}}_{i}(t) \in \mathcal{L}_{\infty}$. As $\tilde{x}_{i}(t)$, $\tilde{r}_{i}(t)$, $\tilde{f}_{i}(t) \in \mathcal{L}_{\infty} \bigcap \mathcal{L}_{2}$ and $\dot{\tilde{x}}_{i}(t)$, $\dot{\tilde{r}}_{i}(t)$, $\dot{\tilde{f}}_{i}(t) \in \mathcal{L}_{\infty}$, Barbalat's lemma \cite{khalil2002nonlinear} can be used to conclude that $\tilde{x}_{i}(t)$, $\tilde{r}_{i}(t)$ and $\tilde{f}_{i}(t)$ converge to zero asymptotically.
\end{proof}

\begin{remark}
	The design of the observer in \eqref{observer} is to illustrate that the consensus algorithm in \eqref{dot_x} is vulnerable to privacy attacks. Various techniques can be exploited to construct the attack model, and it is not the focus of this paper. In the following, we will present a privacy scheme and show that the approach can provide privacy protection against the external eavesdropper no matter what attack model is used.
\end{remark}

\section{Privacy Preservation via State Decomposition} \label{sec_privacy}
%Inspired by the work \cite{WangTAC2019}, a state decomposition scheme is proposed to provide protection against an external eavesdropper. 
Building upon our prior work \cite{WangTAC2019}, in this section, we extend the state decomposition scheme to dynamic average consensus. Note that here we aim at protecting privacy against \emph{any} eavesdropping scheme (including the example shown in Section \ref{sec_eavesdrop}).
More specifically, the state decomposition scheme decomposes the state and reference signals $\left\lbrace x_{i}(t), r_{i}(t), f_{i}(t)\right\rbrace$ of each agent into two sub-sets $\left\lbrace x_{i}^{\alpha}(t), r_{i}^{\alpha}(t), f_{i}^{\alpha}(t) \right\rbrace$ and $\left\lbrace x_{i}^{\beta}(t), r_{i}^{\beta}(t), f_{i}^{\beta}(t) \right\rbrace$. The initial values $r_{i}^{\alpha}(0)$ and $r_{i}^{\beta}(0)$ can be randomly chosen from the set of all real numbers under the following constraint:
\begin{equation} \label{r_alpha_beta}
\begin{aligned}
r_{i}^{\alpha}(0) + r_{i}^{\beta}(0) = 2r_{i}(0).
%&x_{i}^{\alpha}(0) + x_{i}^{\beta}(0) = 2x_{i}(0),
%\\
%&x_{i}^{\alpha}(0) = r_{i}^{\alpha}(0), \quad x_{i}^{\beta}(0) = r_{i}^{\beta}(0).
\end{aligned}
\end{equation}
Furthermore, $f_{i}^{\alpha}(t)$ and $f_{i}^{\beta}(t)$ are bounded and chosen to satisfy:
\begin{equation} \label{f_alpha_beta}
f_{i}^{\alpha}(t) + f_{i}^{\beta}(t) = 2f_{i}(t).
\end{equation}
In this decomposition mechanism, the sub-state $x_{i}^{\alpha}(t)$ takes the role of the original state $x_{i}(t)$ in inter-agent interactions and is the only state value from agent $i$ that will be shared with its neighbors. The other sub-state $x_{i}^{\beta}(t)$ also involves in the distributed updates by (and only by) exchanging information with $x_{i}^{\alpha}(t)$. Therefore, $x_{i}^{\beta}(t)$ will affect the evolution of $x_{i}^{\alpha}(t)$, but the existence of $x_{i}^{\beta}(t)$ is invisible to neighbors of agent $i$ and the eavesdropper. An example is given in Figure~\ref{fig_network} to illustrate the state decomposition of a network with four agents. 
\iffalse
\begin{figure*}[!htbp]
	\centering
	\begin{subfigure}[b]{0.5\textwidth}
		\centering
		\includegraphics[height=1.8 in]{figures/network1.eps}
		\caption{}
		\label{fig_network1}
	\end{subfigure}%\hfill
	\begin{subfigure}[b]{0.5\textwidth}
		\centering
		\includegraphics[height=1.8 in]{figures/network2.eps}
		\caption{}
		\label{fig_network2}
	\end{subfigure}
	\caption{State decomposition: (a) Before state decomposition. (b) After state decomposition.}
	\label{fig_network}
\end{figure*}
\fi
\begin{figure}[!t]
	\centering
	\begin{subfigure}[b]{0.25\textwidth}
		\centering
		\includegraphics[height=1.6 in]{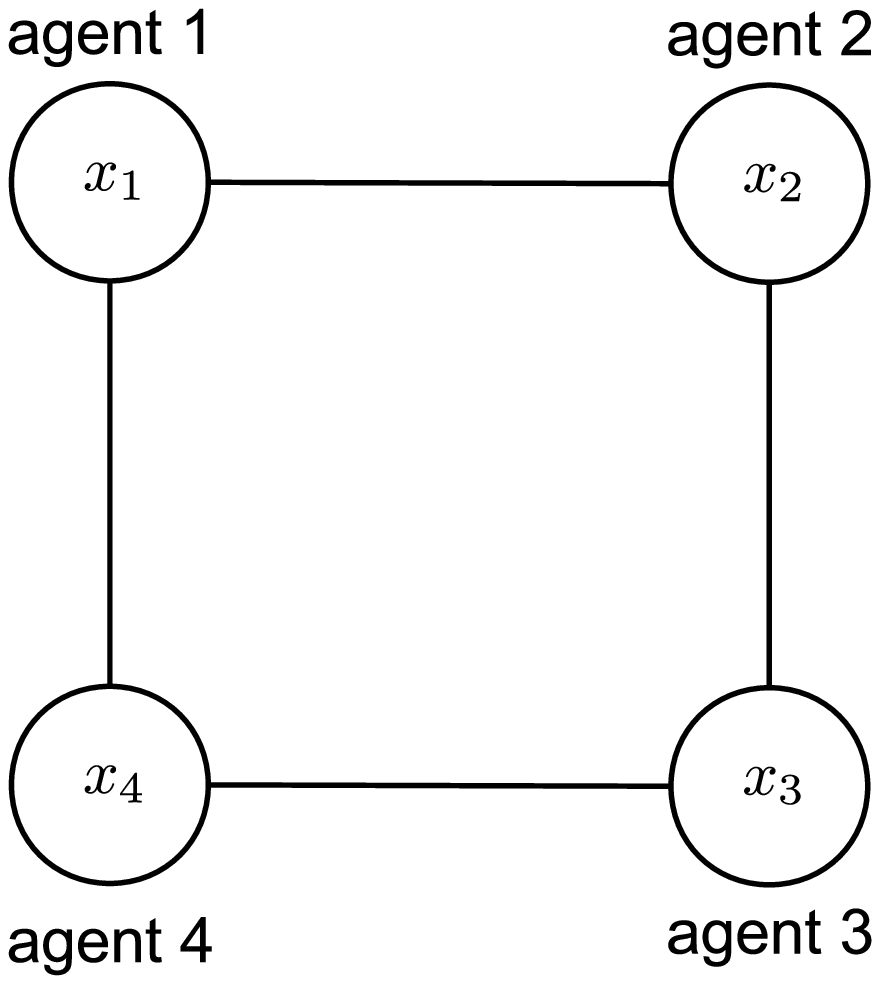}
		\caption{}
		\label{fig_network1}
	\end{subfigure}%\hfill
	\begin{subfigure}[b]{0.25\textwidth}
		\centering
		\includegraphics[height=1.6 in]{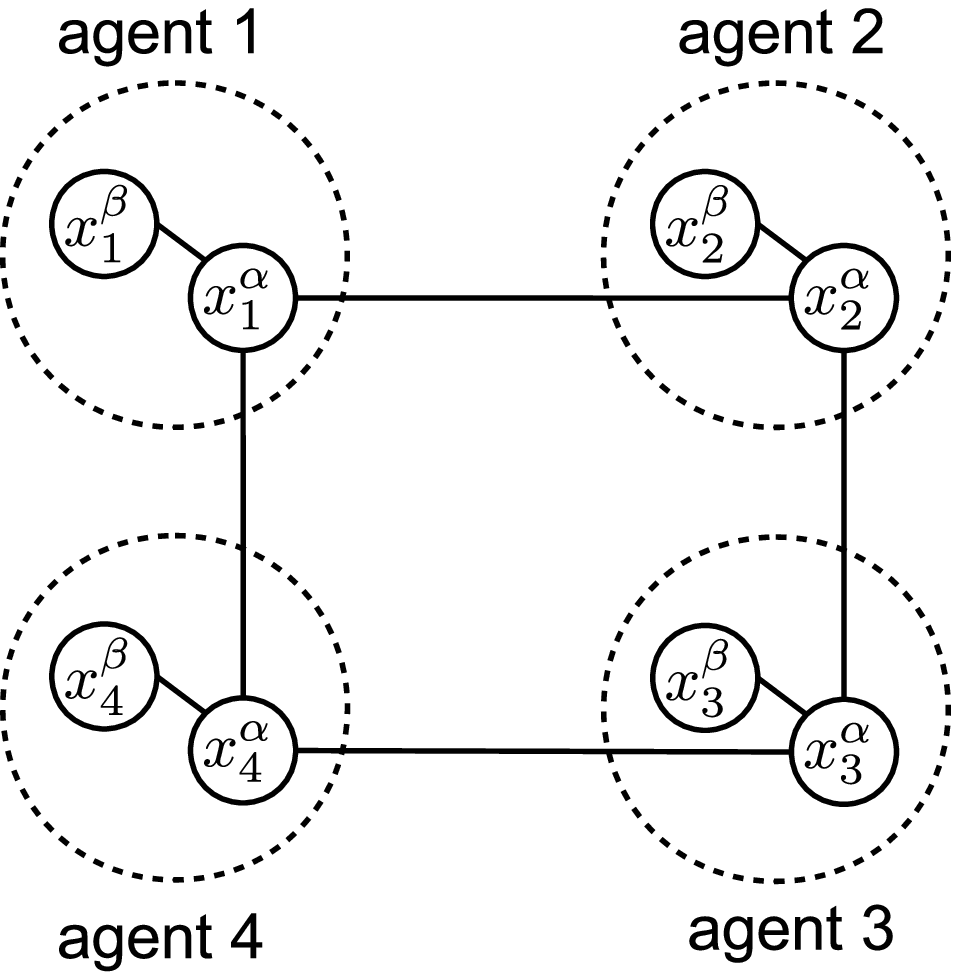}
		\caption{}
		\label{fig_network2}
	\end{subfigure}
	\caption{State decomposition: (a) Before state decomposition. (b) After state decomposition.}
	\label{fig_network}
\end{figure}

Under the decomposition mechanism, the conventional average consensus algorithm in \eqref{dot_x} becomes
\begin{equation} \label{sub_x}
\begin{aligned}
\dot{x}_{i}^{\alpha}(t) &= f_{i}^{\alpha}(t) + \kappa \sum_{j=1}^{n}a_{ij}\left( x_{j}^{\alpha}(t) - x_{i}^{\alpha}(t) \right) 
\\
& \qquad \quad \;\;\, + \kappa \left( x_{i}^{\beta}(t) - x_{i}^{\alpha}(t) \right),
\\
\dot{x}_{i}^{\beta}(t) &= f_{i}^{\beta}(t) + \kappa \left( x_{i}^{\alpha}(t) - x_{i}^{\beta}(t) \right),
\\
x_{i}^{\alpha}(0) &= r_{i}^{\alpha}(0), \quad x_{i}^{\beta}(0) = r_{i}^{\beta}(0), \quad \forall i \in \mathcal{V}.
\end{aligned}
\end{equation}
In the following, we first show that all states $x_{i}^{\alpha}(t)$ and $x_{i}^{\beta}(t)$ will present similar convergence properties as in the conventional case \eqref{dot_x}. Then, we prove that under the decomposition mechanism, the privacy of each agent is protected against an external eavesdropper.

\begin{theorem} \label{theorem_convergence}
	Under the decomposition mechanism, all sub-states $x_{i}^{\alpha}(t)$ and $x_{i}^{\beta}(t)$ in \eqref{sub_x} are input-to-state stable, and the tracking errors $x_{i}^{\alpha}(t)-\frac{1}{n}\sum_{j=1}^{n}r_{j}(t)$ and $x_{i}^{\beta}(t)-\frac{1}{n}\sum_{j=1}^{n}r_{j}(t)$ are ultimately bounded. Moreover, the convergence rate of the tracking errors is no worse than $\frac{\kappa}{2}\left( \lambda_{2}(L)+2-\sqrt{\lambda_{2}^{2}(L)+4} \right)$, where $L$ is the Laplacian matrix of graph before state decomposition and $\lambda_{2}(L) \in \mathbb{R}$ is the second smallest eigenvalue of $L$.
\end{theorem}

\begin{proof}
	It is clear that the decomposition mechanism ensures that all sub-states also compose a connected graph. 
	Based on the results in \cite{KiaCSM2019}, dynamic average consensus can still be achieved, i.e., all sub-states are input-to-state stable, and the convergence errors $x_{i}^{\alpha}(t)-\frac{1}{2n}\sum_{j=1}^{n} \left(r_{j}^{\alpha}(t) + r_{j}^{\beta}(t)\right)$ and $x_{i}^{\beta}(t)-\frac{1}{2n}\sum_{j=1}^{n} \left(r_{j}^{\alpha}(t) + r_{j}^{\beta}(t)\right)$ are ultimately bounded. 
	%In addition, the convergence rate to the error bound is no worse than $\kappa \lambda_{2}(L^{\alpha\beta})$, where $L^{\alpha\beta} \in \mathbb{R}^{2n \times 2n}$ is the Laplacian matrix of the graph composed by all sub-states and $\lambda_{2}(L^{\alpha\beta}) \in \mathbb{R}$ is the second smallest eigenvalue of $L^{\alpha\beta}$.
	%To complete the proof of Theorem \ref{theorem_convergence}, we next show that $\frac{1}{2n}\sum_{j=1}^{n} \left(r_{j}^{\alpha}(t) + r_{j}^{\beta}(t)\right) = \frac{1}{n}\sum_{j=1}^{n}r_{j}(t)$ and $\lambda_{2}(L^{\alpha\beta}) = \frac{1}{2}\left( \lambda_{2}(L)+2-\sqrt{\lambda_{2}^{2}(L)+4} \right)$.
	It can be obtained from \eqref{r_alpha_beta} and \eqref{f_alpha_beta} that 
	\begin{equation}
	r_{i}^{\alpha}(t) + r_{i}^{\beta}(t) = 2r_{i}(t),
	\end{equation}
	which implies $\frac{1}{2n}\sum_{j=1}^{n} \left(r_{j}^{\alpha}(t) + r_{j}^{\beta}(t)\right) = \frac{1}{n}\sum_{j=1}^{n}r_{j}(t)$. Therefore, all sub-states $x_{i}^{\alpha}(t)$ and $x_{i}^{\beta}(t)$ in \eqref{sub_x} retain the convergence to the neighborhood of the same average consensus value as the original states. %will converge to the dynamic average consensus value $\frac{1}{n}\sum_{j=1}^{n}r_{j}(t)$.
	
	Let $L^{\alpha\beta} \in \mathbb{R}^{2n \times 2n}$ be the Laplacian matrix of the graph composed by all sub-states. By leveraging the results in \cite{KiaCSM2019}, it can be concluded that the convergence rate of the consensus algorithm is no worse than $\kappa \lambda_{2}(L^{\alpha\beta})$, where $\lambda_{2}(L^{\alpha\beta}) \in \mathbb{R}$ is the second smallest eigenvalue of $L^{\alpha\beta}$. To complete the proof of Theorem \ref{theorem_convergence}, we next show that $\lambda_{2}(L^{\alpha\beta}) = \frac{1}{2}\left( \lambda_{2}(L)+2-\sqrt{\lambda_{2}^{2}(L)+4} \right)$. According to the decomposition mechanism, $L^{\alpha\beta}$ can be formulated as
	\begin{equation} \label{L_alpha_beta}
	L^{\alpha\beta} = \begin{bmatrix}
	L+I_{n} & -I_{n}
	\\
	-I_{n} & I_{n}
	\end{bmatrix},
	\end{equation}
	with $I_{n} \in \mathbb{R}^{n\times n}$ being the $n$-dimensional identity matrix. Let $\lambda(L_{\alpha\beta})$, $\lambda(L) \in \mathbb{R}$ be the eigenvalue of $L_{\alpha\beta}$ and $L$, respectively. Based on \eqref{L_alpha_beta} and the eigenvalue-eigenvector equation \cite{hornBOOK2012}, it can be derived that 
	\begin{equation} \label{lambda_L}
	\lambda(L_{\alpha\beta}) = \frac{1}{2}\left( \lambda(L)+2\pm\sqrt{\lambda^{2}(L)+4} \right).
	\end{equation} 
	From \eqref{lambda_L}, it can be further deduced that the second smallest eigenvalue of $L_{\alpha\beta}$, i.e., $\lambda_{2}(L^{\alpha\beta})$, is $\frac{1}{2}\left( \lambda_{2}(L)+2-\sqrt{\lambda_{2}^{2}(L)+4} \right)$, which completes the proof.
\end{proof}

\begin{remark}
	Given a graph topology, the second smallest eigenvalue of Laplacian $\lambda_{2}(\cdot)$ is a measure of the convergence rate of consensus algorithms \cite{Olfati-Saber2007}. The convergence rate of the algorithm in \eqref{dot_x} is no worse than $\kappa \lambda_{2}(L)$, and after state decomposition, this measure will decrease to $\frac{\kappa}{2}\left( \lambda_{2}(L)+2-\sqrt{\lambda_{2}^{2}(L)+4} \right)$ since the connectivity of the graph becomes weaker. 
	%Nevertheless, under the decomposition mechanism, all the sub-states will still 
\end{remark}

We next show that the decomposition scheme can protect the privacy of the agents against the external eavesdropper.

\begin{theorem} \label{theorem_privacy}
	Under the decomposition mechanism, an external eavesdropper cannot infer the reference signals $r_{p}(t)$ and $f_{p}(t)$ of any agent $p$ with any guaranteed accuracy.
\end{theorem}

\begin{proof}
	Under the decomposition mechanism, the information accessible to the eavesdropper at time $t$ can be defined as $I(t) \triangleq \left\lbrace A, \kappa, x_{i}^{\alpha}(t), i=1, 2, \cdots, n \right\rbrace$. To show that the privacy of the reference signals $r_{p}(t)$ and $f_{p}(t)$ can be preserved against the eavesdropper, it suffices to present that under any value $\bar{r}_{p}(t) = \bar{r}_{p}(0) + \int_{0}^{t} \bar{f}_{p}(\tau)d\tau$ satisfying $\bar{r}_{p}(t) \neq r_{p}(t)$, the information $\bar{I}(t) \triangleq \left\lbrace A, \kappa, \bar{x}_{i}^{\alpha}(t), i=1, 2, \cdots, n \right\rbrace$ accessible to the eavesdropper could be exactly the same as the information $I(t)$ cumulated under $r_{p}(t)$. This is because the only information available for the eavesdropper to extract the signals $r_{p}(t)$ and $f_{p}(t)$ is $I(t)$, and if $I(t)$ could be the outcome under any value of $r_{p}(t)$, then the eavesdropper has no way to even find a range for $r_{p}(t)$ and $f_{p}(t)$. Therefore, we only need to prove that for any $\bar{r}_{p}(t) \neq r_{p}(t)$, $\bar{I}(t) = I(t)$ could hold.
	
	Let agent $l$ be one of the neighbors of agent $p$. Next we show that given $\bar{r}_{p}(t)$ (i.e., $\bar{r}_{p}(0)$ and $\bar{f}_{p}(t)$), by suitably selecting the values of $\bar{r}_{l}(0)$, $\bar{r}_{p}^{\alpha}(0)$, $\bar{r}_{p}^{\beta}(0)$, $\bar{r}_{l}^{\alpha}(0)$, $\bar{r}_{l}^{\beta}(0)$, $\bar{f}_{p}^{\alpha}(t)$, $\bar{f}_{p}^{\beta}(t)$, $\bar{f}_{l}^{\alpha}(t)$, and $\bar{f}_{l}^{\beta}(t)$, $\bar{I}(t) = I(t)$ could hold under $\bar{r}_{p}(t) \neq r_{p}(t)$.
	%Next we show that given $\bar{r}_{p}(t)$ (i.e., $\bar{r}_{p}(0)$ and $\bar{f}_{p}(t)$), by suitably selecting the values of $\bar{r}_{p}^{\alpha}(0)$, $\bar{r}_{p}^{\beta}(0)$, $\bar{f}_{p}^{\alpha}(t)$, and $\bar{f}_{p}^{\beta}(t)$, $\bar{I}(t) = I(t)$ can hold under $\bar{r}_{p}(t) \neq r_{p}(t)$. 
	Specifically, under the following conditions:
	\begin{equation} \label{bar_r0}
	\begin{aligned}
	\bar{r}_{l}(0) &= r_{l}(0) + r_{p}(0) - \bar{r}_{p}(0),
	\\
	\bar{r}_{p}^{\alpha}(0) &= r_{p}^{\alpha}(0), \bar{r}_{p}^{\beta}(0) = 2\bar{r}_{p}(0) - r_{p}^{\alpha}(0),
	\\
	\bar{r}_{l}^{\alpha}(0) &= r_{l}^{\alpha}(0), \bar{r}_{l}^{\beta}(0) = 2\bar{r}_{l}(0) - r_{l}^{\alpha}(0),
	\\
	\bar{r}_{q}(0) &= r_{q}(0), \bar{r}_{q}^{\alpha}(0) = r_{q}^{\alpha}(0), \bar{r}_{q}^{\beta}(0) = r_{q}^{\beta}(0), \forall q \in \mathcal{V}\setminus \left\lbrace p, l \right\rbrace,
	\end{aligned}
	\end{equation}	
	\begin{equation} \label{bar_f}
	\begin{aligned}
	\bar{f}_{l}(t) &= f_{l}(t) + f_{p}(t) - \bar{f}_{p}(t),
	\\
	\bar{f}_{p}^{\alpha}(t) &= f_{p}^{\alpha}(t) + 2\kappa \left( r_{p}(t)-\bar{r}_{p}(t) \right), \bar{f}_{p}^{\beta}(t) = 2\bar{f}_{p}(t) - \bar{f}_{p}^{\alpha}(t),
	\\
	\bar{f}_{l}^{\alpha}(t) &= f_{l}^{\alpha}(t) + 2\kappa \left( \bar{r}_{p}(t)-r_{p}(t) \right), \bar{f}_{l}^{\beta}(t) = 2\bar{f}_{l}(t) - \bar{f}_{l}^{\alpha}(t),
	\\
	\bar{f}_{q}(t) &= f_{q}(t), \bar{f}_{q}^{\alpha}(t) = f_{q}^{\alpha}(t), \bar{f}_{q}^{\beta}(t) = f_{q}^{\beta}(t), \forall q \in \mathcal{V}\setminus \left\lbrace p, l \right\rbrace,
	\end{aligned}
	\end{equation}
	and system dynamics
	\begin{equation} \label{sub_bar_x}
	\begin{aligned}
	\dot{\bar{x}}_{i}^{\alpha}(t) &= \bar{f}_{i}^{\alpha}(t) + \kappa \sum_{j=1}^{n}a_{ij}\left( \bar{x}_{j}^{\alpha}(t) - \bar{x}_{i}^{\alpha}(t) \right) 
	\\
	& \qquad \quad \;\;\, + \kappa \left( \bar{x}_{i}^{\beta}(t) - \bar{x}_{i}^{\alpha}(t) \right),
	\\
	\dot{\bar{x}}_{i}^{\beta}(t) &= \bar{f}_{i}^{\beta}(t) + \kappa \left( \bar{x}_{i}^{\alpha}(t) - \bar{x}_{i}^{\beta}(t) \right),
	\\
	\bar{x}_{i}^{\alpha}(0) &= \bar{r}_{i}^{\alpha}(0), \quad \bar{x}_{i}^{\beta}(0) = \bar{r}_{i}^{\beta}(0), \quad \forall i \in \mathcal{V},
	\end{aligned}
	\end{equation}
	the new sub-state $\bar{x}_{i}^{\alpha}(t)$ will be the same as $x_{i}^{\alpha}(t)$, for all $i \in \mathcal{V}$, i.e., $\bar{I}(t) = I(t)$. 
	%Based on \eqref{bar_r0} and \eqref{bar_f}, it is clear that if $\bar{x}_{p}^{\alpha}(t) = x_{p}^{\alpha}(t)$, then $\forall i \in\mathcal{V}$ 
	Note that the first equations in both \eqref{bar_r0} and \eqref{bar_f} are introduced to ensure that the average consensus value $\frac{1}{n}\sum_{j=1}^{n} \bar{r}_{j}(t)$ is the same as the original one $\frac{1}{n}\sum_{j=1}^{n}r_{j}(t)$.
	Now we prove that $\bar{I}(t) = I(t)$. First, from \eqref{r_alpha_beta}, \eqref{sub_x}, \eqref{bar_r0}, and \eqref{sub_bar_x}, it can be obtained that
	\begin{equation} \label{bar_x0}
	\begin{aligned}
	\bar{x}_{p}^{\alpha}(0) &= x_{p}^{\alpha}(0), \bar{x}_{p}^{\beta}(0) = x_{p}^{\beta}(0) + 2\left( \bar{r}_{p}(0)-r_{p}(0) \right),
	\\
	\bar{x}_{l}^{\alpha}(0) &= x_{l}^{\alpha}(0), \bar{x}_{l}^{\beta}(0) = x_{l}^{\beta}(0) + 2\left( r_{p}(0)-\bar{r}_{p}(0) \right),
	\\
	\bar{x}_{q}^{\alpha}(0) &= x_{q}^{\alpha}(0), \bar{x}_{q}^{\beta}(0) = x_{q}^{\beta}(0), \forall q \in \mathcal{V}\setminus \left\lbrace p, l \right\rbrace.
	\end{aligned}
	\end{equation}
	Furthermore, based on \eqref{bar_f} and \eqref{bar_x0}, it can be verified that
	\begin{equation} \label{bar_xt}
	\begin{aligned}
	\bar{x}_{p}^{\alpha}(t) &= x_{p}^{\alpha}(t), \bar{x}_{p}^{\beta}(t) = x_{p}^{\beta}(t) + 2\left( \bar{r}_{p}(t)-r_{p}(t) \right),
	\\
	\bar{x}_{l}^{\alpha}(t) &= x_{l}^{\alpha}(t), \bar{x}_{l}^{\beta}(t) = x_{l}^{\beta}(t) + 2\left( r_{p}(t)-\bar{r}_{p}(t) \right),
	\\
	\bar{x}_{q}^{\alpha}(t) &= x_{q}^{\alpha}(t), \bar{x}_{q}^{\beta}(t) = x_{q}^{\beta}(t), \forall q \in \mathcal{V}\setminus \left\lbrace p, l \right\rbrace, 
	\end{aligned}
	\end{equation}
	is the solution to \eqref{sub_bar_x}. It is obvious that the solution \eqref{bar_xt} satisfies $\forall i \in \mathcal{V}$, $\bar{x}_{i}^{\alpha}(t) = x_{i}^{\alpha}(t)$, and thus $\bar{I}(t) = I(t)$ could hold under $\bar{r}_{p}(t) \neq r_{p}(t)$. Based on the above analysis, it can be concluded that no matter what attack model is used, the eavesdropper cannot infer $r_{p}(t)$ and $f_{p}(t)$ from $I(t)$ with any guaranteed accuracy.
\end{proof}

\begin{remark}
	The proposed state decomposition scheme is a general privacy-preserving augmentation to dynamic average consensus algorithms. 
	%Since the scheme is completely decentralized, 
	While the above analysis is based on the dynamic consensus algorithm in \eqref{dot_x}, the state decomposition framework can be integrated with other dynamic average consensus approaches (e.g., \cite{Freeman2006CDC,Bai2010CDC,ChenTAC2012,ChenTAC2015}) to improve the resilience of original methods to privacy attacks.
\end{remark}

\begin{remark}
	The proposed state decomposition scheme is a significant extension on the basis of the work \cite{WangTAC2019}. Different from \cite{WangTAC2019} that focuses on protecting the privacy in discrete-time static average consensus, we exploit the decomposition mechanism to achieve the privacy preservation in continuous-time dynamic average consensus. As dynamic average consensus involves  dynamically evolving reference signals, more sophisticated analytical tools relying on algebraic graph theory and control theory are used to conduct the development of the proposed scheme.  
\end{remark}

\iffalse
\begin{theorem}
	Under the decomposition mechanism, an external eavesdropper cannot infer the reference signals $r_{i}(t)$ and $f_{i}(t)$ of any agent $i$ with any guaranteed accuracy.
\end{theorem}

\begin{IEEEproof}
	To show that the eavesdropper cannot identify $r_{i}(t)$ with any guaranteed accuracy, we present that any variation of $r_{i}(0)$ is indistinguishable to the eavesdropper, i.e., the information accessible to the eavesdropper can be the same even if $r_{i}(0)$ is changed to arbitrary values $\bar{r}_{i}(0) \leq r_{i}(0)$. The information accessible to the eavesdropper at time $t$ is defined as $I(t) \triangleq \left\lbrace A, \kappa, x_{i}^{\alpha}(t), i=1, 2, \cdots, n \right\rbrace$. The initial value
\end{IEEEproof}
\fi

\section{Application to Formation Control} \label{sec_formation}

In this section, the privacy-preserving dynamic average consensus is applied to the formation control of non-holonomic mobile robots under the adversarial environment with eavesdropping attackers.

\subsection{Formation Control Objective}

\begin{figure}[!htbp]
	\centering
	\hspace{0 in}
	\includegraphics[width=2.5 in]{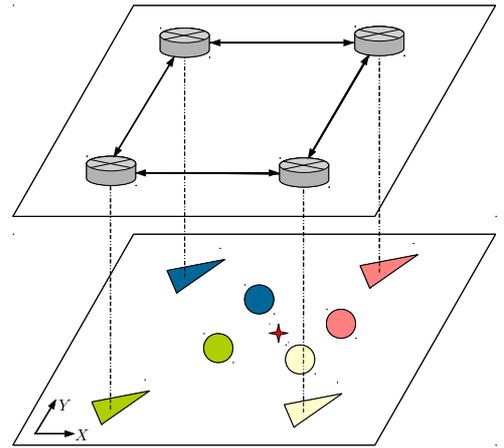}
	\caption{Dynamic consensus based formation control. The triangles are the non-holonomic mobile robots; the circles are the mobile targets; and the cross is the center of the mobile targets.}
	\label{fig_formation}
\end{figure}
In particular, consider  $n$ networked non-holonomic mobile robots and $n$ mobile moving targets in the motion plane. Each mobile robot $i$ has access to its own position $s_{i}(t) \triangleq \begin{bmatrix}
s_{xi}(t) & s_{yi}(t)
\end{bmatrix}^{T} \in \mathbb{R}^{2}$ and can monitor the mobile target $i$'s position $p_{i}(t)\triangleq \begin{bmatrix}
p_{xi}(t) & p_{yi}(t)
\end{bmatrix}^{T} \in \mathbb{R}^{2}$ and velocity $\dot{p}_{i}(t)=q_{i}(t)\triangleq \begin{bmatrix}
q_{xi}(t) & q_{yi}(t)
\end{bmatrix}^{T} \in \mathbb{R}^{2}$. $s_{i}(t)$, $p_{i}(t)$, and $q_{i}(t)$ are all expressed with respect to the inertial coordinate frame. 
Each mobile robot shares relevant information with its neighbors via wireless communication. 
We consider the case that the mobile robots and targets are operating in the adversarial environment with eavesdropping attackers.
The objective of the networked mobile robots is to follow the set of mobile targets by spreading out in a pre-specified formation while preserving the privacy of mobile targets against the eavesdropper. Specifically, the formation control requires that by using the information received from the communication network, each mobile robot $i$ is driven to a relative vector $b_{i}(t) \triangleq \begin{bmatrix}
b_{xi}(t) & b_{yi}(t)
\end{bmatrix}^{T} \in \mathbb{R}^{2}$ with respect to the time-varying geometric center of the mobile targets, i.e., $s_{i}(t) \rightarrow \frac{1}{n} \sum_{i=1}^{n}p_{i}(t) + b_{i}(t)$ as $t \rightarrow \infty$. Meanwhile, since each mobile robot can only monitor one mobile target, it needs to cooperate with its neighbors to compute the geometric center, which induces the risk of breaching the privacy of mobile targets. An external eavesdropper can use the leaked sensitive information, such as the position and/or velocity information, to maliciously track and attack the mobile targets. Therefore,
it is important that the formation control scheme can protect the privacy of mobile targets against the eavesdropping attacker, i.e., an external eavesdropper cannot identify the position and/or velocity of mobile targets based on the network information. An example scenario in which a team of mobile robots tracks a group of mobile targets is depicted in Figure~\ref{fig_formation}.

\subsection{Control Design}
As discussed in \cite{Porfiri2007AUTO,KiaCSM2019}, the formation control problem in the aforementioned scenario can be addressed with a two-layer method. In the cyber layer, the privacy-preserving dynamic average consensus algorithm is used to estimate the geometric center of mobile targets in a distributed manner, while in the physical layer, the mobile robot $i$ is actuated to follow the estimate of the geometric center with a desired relative bias $b_{i}(t)$. The implementation details are given now.

%For each mobile robot $i$, let $c_{i}(t)\triangleq \begin{bmatrix} c_{xi}(t) & c_{yi}(t) \end{bmatrix}^{T} \in \mathbb{R}^{2}$ be the estimate of the geometric center of mobile targets. 
%Since each mobile robot only can monitor one mobile target, it needs to cooperate with its neighbors to compute the geometric center in the cyber layer. 
In the cyber layer, the state decomposition based dynamic average consensus algorithm presented in Section \ref{sec_privacy} is utilized for the calculation of the geometric center $\frac{1}{n} \sum_{i=1}^{n}p_{i}(t)$ and the privacy preservation of mobile targets. More precisely, let $c_{i}^{\alpha}(t) \triangleq \begin{bmatrix}
c_{xi}^{\alpha}(t) & c_{yi}^{\alpha}(t)
\end{bmatrix}^{T} \in \mathbb{R}^{2}$ and $c_{i}^{\beta}(t) \triangleq \begin{bmatrix}
c_{xi}^{\beta}(t) & c_{yi}^{\beta}(t)
\end{bmatrix}^{T} \in \mathbb{R}^{2}$ be two sub-sets of the geometric center estimates, which are updated by
\begin{equation} \label{sub_c}
\begin{aligned}
\dot{c}_{i}^{\alpha}(t) &= q_{i}^{\alpha}(t) + \kappa \sum_{j=1}^{n}a_{ij}\left( c_{j}^{\alpha}(t) - c_{i}^{\alpha}(t) \right) 
\\
& \qquad \quad \;\;\, + \kappa \left( c_{i}^{\beta}(t) - c_{i}^{\alpha}(t) \right),
\\
\dot{c}_{i}^{\beta}(t) &= q_{i}^{\beta}(t) + \kappa \left( c_{i}^{\alpha}(t) - c_{i}^{\beta}(t) \right),
\\
c_{i}^{\alpha}(0) &= p_{i}^{\alpha}(0), \quad c_{i}^{\beta}(0) = p_{i}^{\beta}(0),
\end{aligned}
\end{equation} 
where $p_{i}^{\alpha}(0)$, $p_{i}^{\beta}(0)$, $q_{i}^{\alpha}(t)$, $q_{i}^{\beta}(t) \in \mathbb{R}^{2}$ are selected to satisfy
\begin{equation} \label{p_alpha_beta}
\begin{aligned}
p_{i}^{\alpha}(0) + p_{i}^{\beta}(0) = 2p_{i}(0),
\\
q_{i}^{\alpha}(t) + q_{i}^{\beta}(t) = 2q_{i}(t).
\end{aligned}
\end{equation}
As shown in Theorem \ref{theorem_convergence} and Theorem \ref{theorem_privacy}, $c_{i}^{\alpha}(t)$ will converge to the neighbourhood of the geometric center $\frac{1}{n} \sum_{i=1}^{n}p_{i}(t)$, and the mobile targets' information cannot be identified by the external eavesdropper.

In the physical layer, the objective now is to design a tracking controller for non-holonomic mobile robot $i$ to ensure that $s_{i}(t) \rightarrow c_{i}^{\alpha}(t) + b_{i}(t)$ as $t \rightarrow \infty$. The kinematic model of non-holonomic mobile robot $i$ is described by 
\begin{equation} \label{dot_si}
\begin{aligned}
\dot{s}_{xi}(t) &= v_{i}(t)\cos(\theta_{i}(t)),
\\
\dot{s}_{yi}(t) &= v_{i}(t)\sin(\theta_{i}(t)),
\\
\dot{\theta}_{i}(t) &= \omega_{i}(t),
\end{aligned}
\end{equation}
where $\theta_{i}(t) \in \mathbb{R}$ is the heading angle expressed in the inertial coordinate frame, and $v_{i}(t)$, $\omega_{i}(t) \in \mathbb{R}$ are the linear and angular velocity, respectively. To facilitate the following development, the desired heading angle $\theta_{di}(t)\in \mathbb{R}$ and desired linear velocity $v_{di}(t)\in \mathbb{R}$ are constructed as
\begin{equation} \label{thetadi_vdi}
\begin{aligned}
\theta_{di}(t) &= \arctan \left(\frac{\dot{c}_{yi}^{\alpha}(t)}{\dot{c}_{xi}^{\alpha}(t)} \right),
\\
v_{di}(t) &= \sqrt{(\dot{c}_{xi}^{\alpha}(t))^{2} + (\dot{c}_{yi}^{\alpha}(t))^{2}},
\end{aligned}
\end{equation}
which indicates that $\dot{c}_{xi}(t)$ and $\dot{c}_{yi}(t)$ can be rewritten as $\dot{c}_{xi}(t) = v_{di}(t)\cos(\theta_{di}(t))$, $\dot{c}_{yi}(t) = v_{di}(t)\sin(\theta_{di}(t))$. Based on coordinate transformation, the system errors are defined as
\begin{equation} \label{ei}
\begin{aligned}
e_{xi}(t) &\triangleq \cos(\theta_{i}(t)) (s_{xi}(t)-c_{xi}^{\alpha}(t)-b_{xi}(t)) 
\\
& \quad \, + \sin(\theta_{i}(t)) (s_{yi}(t)-c_{yi}^{\alpha}(t)-b_{yi}(t)),
\\
e_{yi}(t) &\triangleq -\sin(\theta_{i}(t)) (s_{xi}(t)-c_{xi}^{\alpha}(t)-b_{xi}(t)) 
\\
& \quad \, + \cos(\theta_{i}(t)) (s_{yi}(t)-c_{yi}^{\alpha}(t)-b_{yi}(t)),
\\
e_{\theta i}(t) &\triangleq \theta_{i}(t) - \theta_{di}(t).
\end{aligned}
\end{equation}
It is clear that $s_{i}(t) \rightarrow c_{i}^{\alpha}(t) + b_{i}(t)$ as $\begin{bmatrix}
e_{xi}(t) & e_{yi}(t) & e_{\theta i}(t)
\end{bmatrix} \rightarrow 0$. Note that the mobile robot is subjected to non-holonomic constraint, and thus in general time-varying auxiliary variables are needed to facilitate the controller design \cite{Wang2015Simultaneous,Huang2013234,JIANGdagger19971393}. Considering the non-holonomic constraint, an auxiliary error $\bar{e}_{\theta i}(t) \in \mathbb{R}$ is defined as
\begin{equation} \label{bar_ei}
\bar{e}_{\theta i}(t) \triangleq e_{\theta i}(t) - \rho_{i}(t),
\end{equation}
where the time-varying signal $\rho_{i}(t) \in \mathbb{R}$ is given by
\begin{equation} \label{rho_i}
\begin{aligned}
\rho_{i}(t) &\triangleq \iota_{0} \varpi_{i}(t) \tanh \left(\iota_{1}\sqrt{e_{xi}^{2}(t) + e_{yi}^{2}(t)} \right) \sin(\iota_{2}t)
\end{aligned}
\end{equation}
with $\varpi_{i}(t) \triangleq \exp \left(-\int_{0}^{t} |v_{di}(\tau)|d\tau \right) \in \mathbb{R} $ and $\iota_{0}$, $\iota_{1}$, $\iota_{2} \in \mathbb{R}$ being positive constants. To achieve the formation control, the velocity inputs $v_{i}(t)$ and $\omega_{i}(t)$ are designed as
\begin{equation} \label{c_input}
\begin{aligned}
v_{i}(t) &= -\gamma_{1}\tanh(e_{xi}(t)) + \cos(e_{\theta i}(t)) v_{di}(t),
\\
\omega_{i}(t) &= -\gamma_{2}\tanh(\bar{e}_{\theta i}(t)) + \dot{\rho}_{i}(t) - \gamma_{3} \text{sgn}(\bar{e}_{\theta i}(t))
\\
& \quad \, - \gamma_{4}\frac{\sin(e_{\theta i}(t)) - \sin(\rho_{i}(t))}{\bar{e}_{\theta i}(t)} v_{di}(t) e_{yi}(t),
\end{aligned}
\end{equation}
where $\gamma_{1}$, $\gamma_{2}$, $\gamma_{3}$, $\gamma_{4} \in \mathbb{R}$ are positive control gains, and $\text{sgn}(\cdot)$ is the standard signum function.

\begin{figure*}[!h]
	\centering
	\begin{subfigure}[b]{0.33\textwidth}
		\centering
		\includegraphics[width=2.2 in]{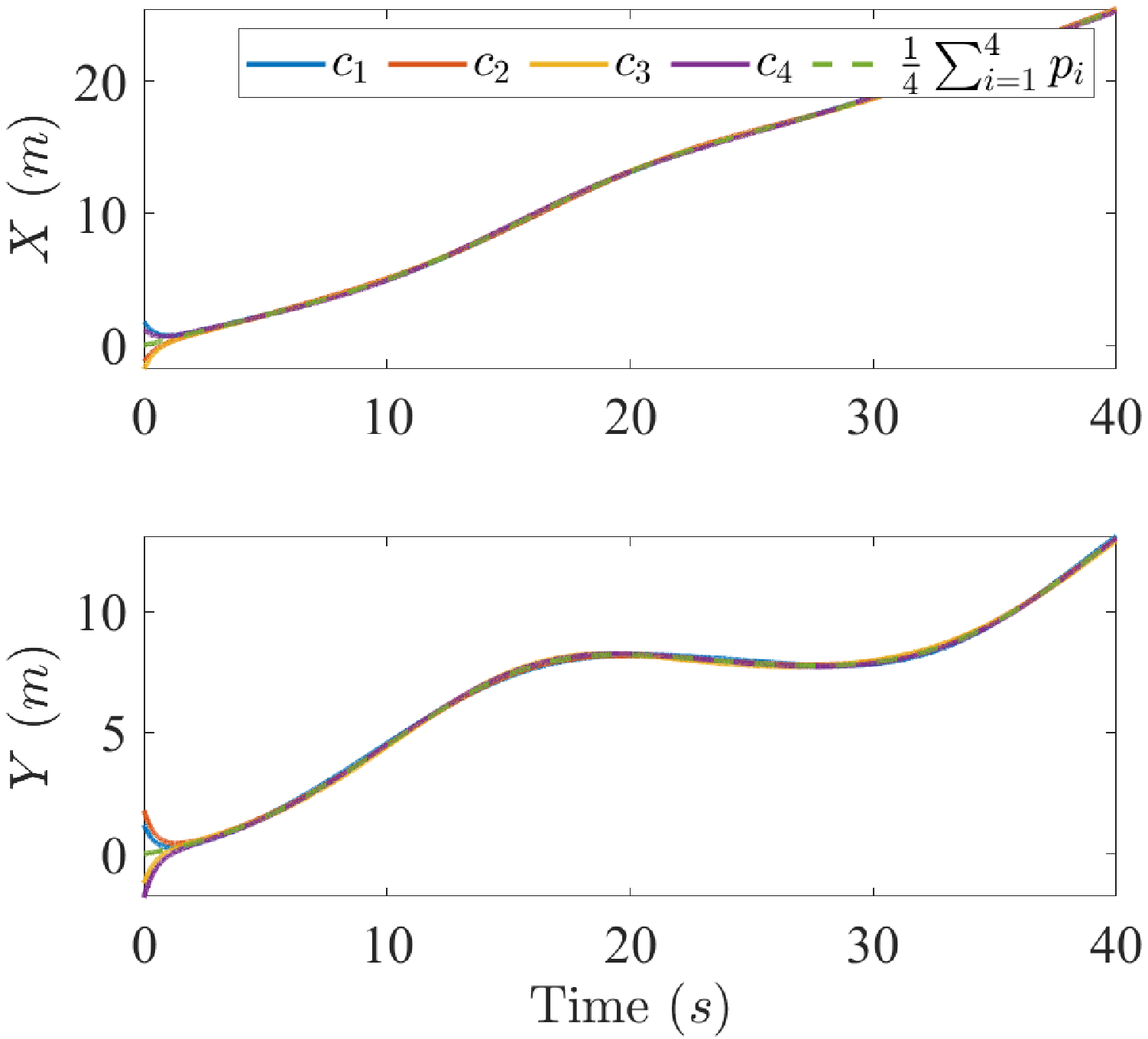}
		\caption{}
		\label{fig_consensus_c}
	\end{subfigure}
	\begin{subfigure}[b]{0.33\textwidth}
		\centering
		\includegraphics[width=2.2 in]{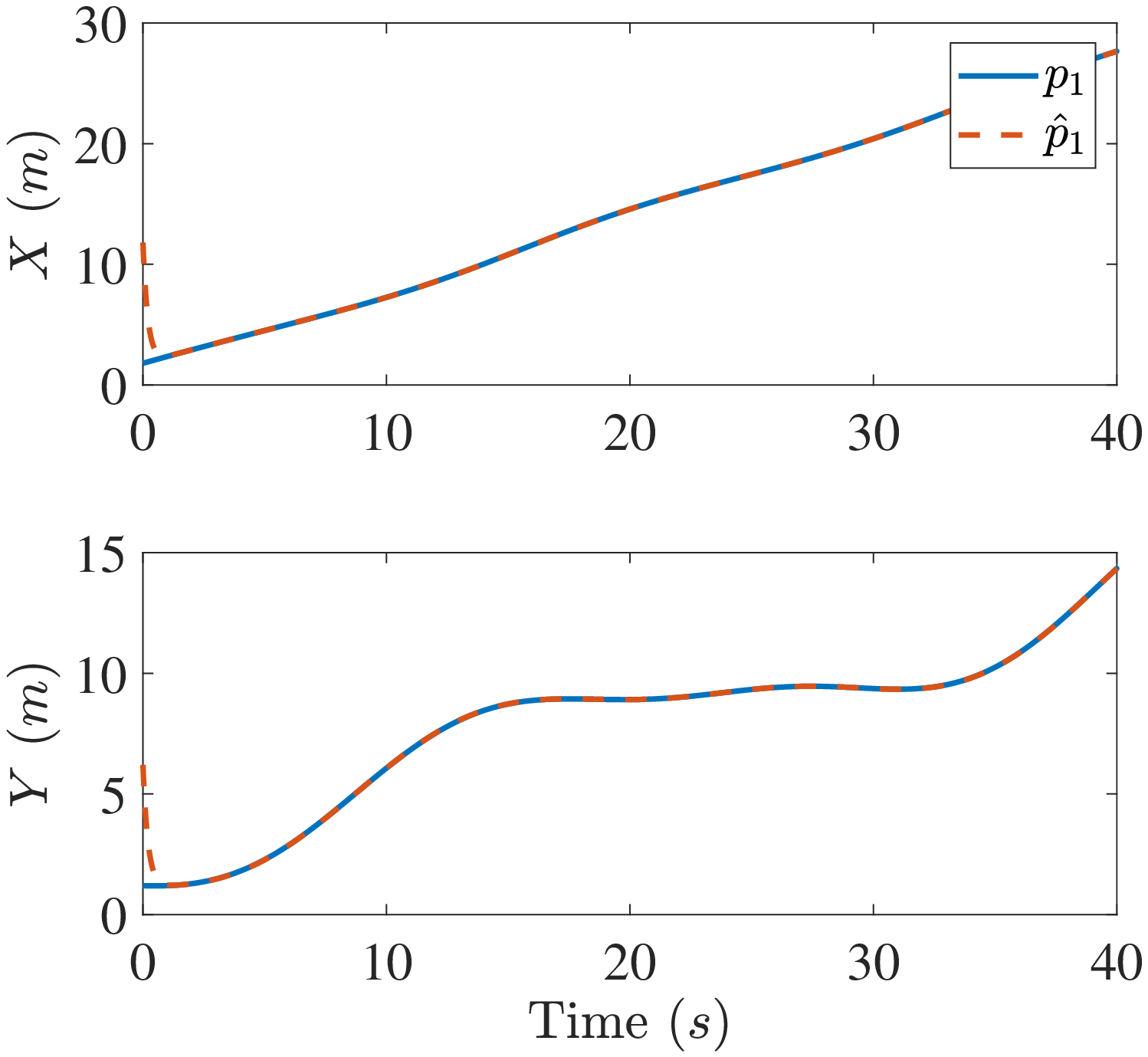}
		\caption{}
		\label{fig_consensus_p}
	\end{subfigure}
	\begin{subfigure}[b]{0.33\textwidth}
		\centering
		\includegraphics[width=2.2 in]{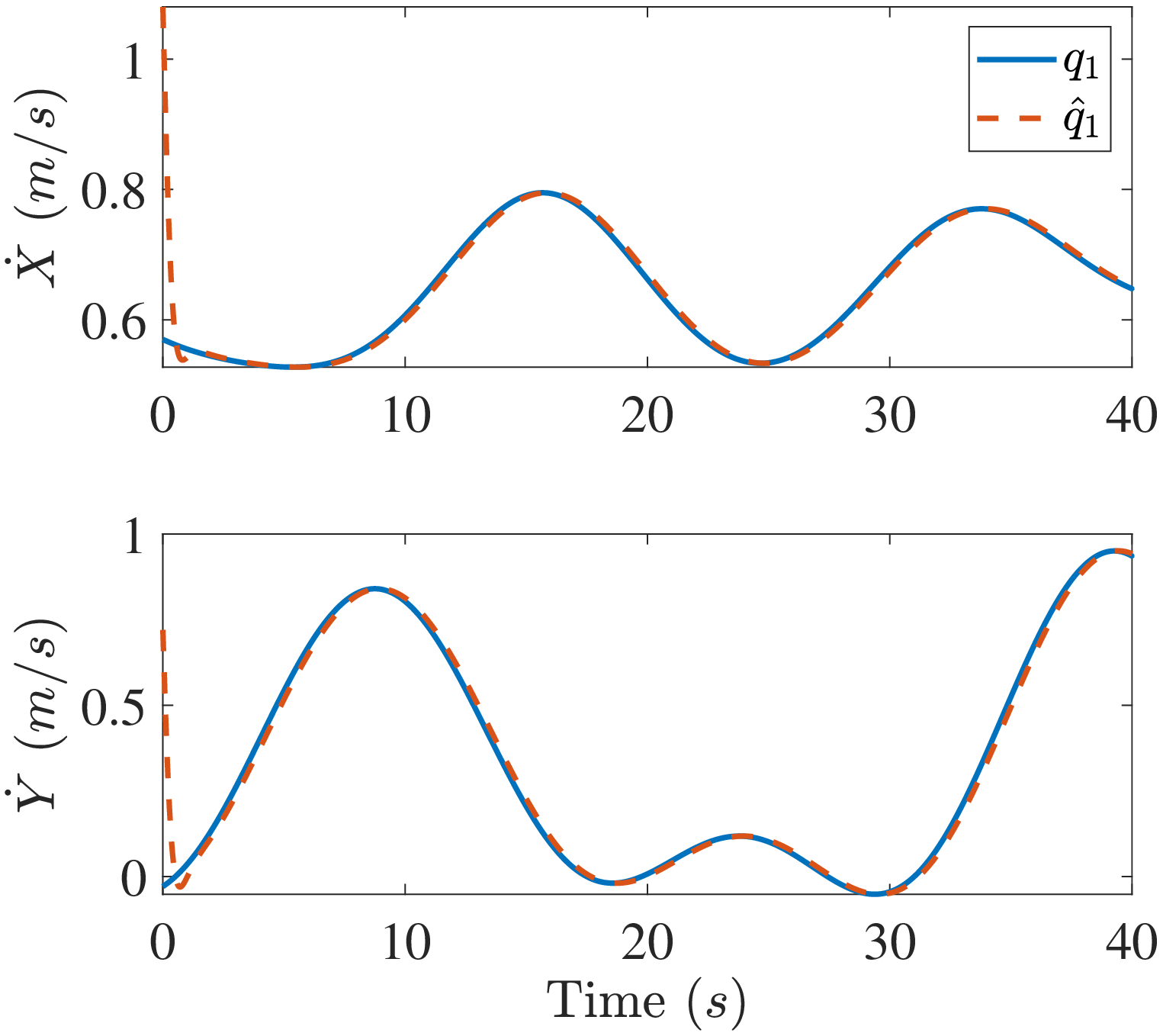}
		\caption{}
		\label{fig_consensus_q}
	\end{subfigure}
	\caption{Simulation results of conventional dynamic average consensus \eqref{dot_x}: (a) System state convergence. (b) Estimation of $p_{1}(t)$ with the eavesdropping scheme developed in Section \ref{sec_eavesdrop}. (c) Estimation of $q_{1}(t)$ with the eavesdropping scheme developed in Section \ref{sec_eavesdrop}.}
	\label{fig_sim1}
\end{figure*}
\begin{figure*}[!htbp]
	\centering
	\begin{subfigure}[b]{0.33\textwidth}
		\centering
		\includegraphics[width=2.2 in]{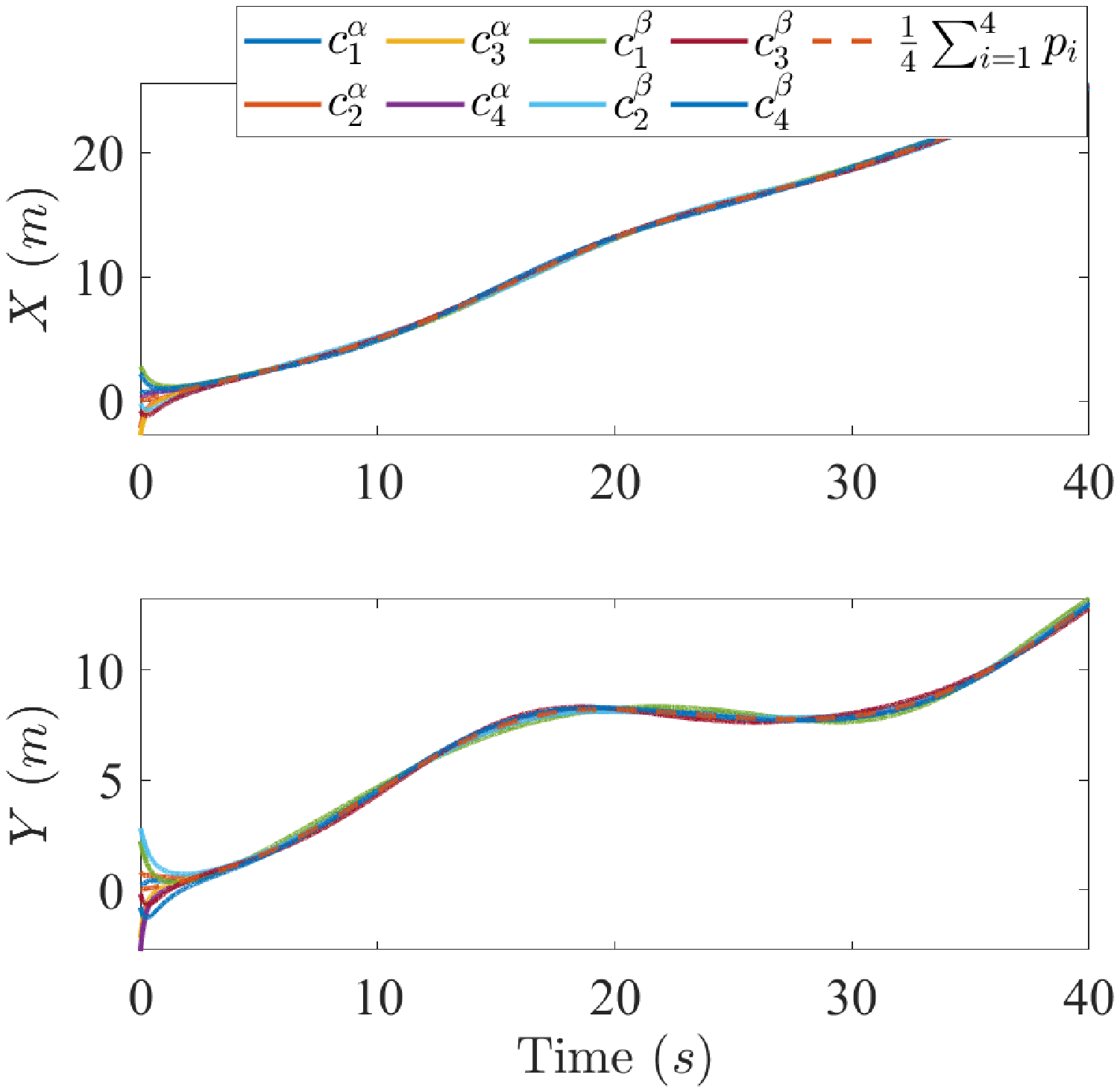}
		\caption{}
		\label{fig_consensus_c_privacy}
	\end{subfigure}
	\begin{subfigure}[b]{0.33\textwidth}
		\centering
		\includegraphics[width=2.2 in]{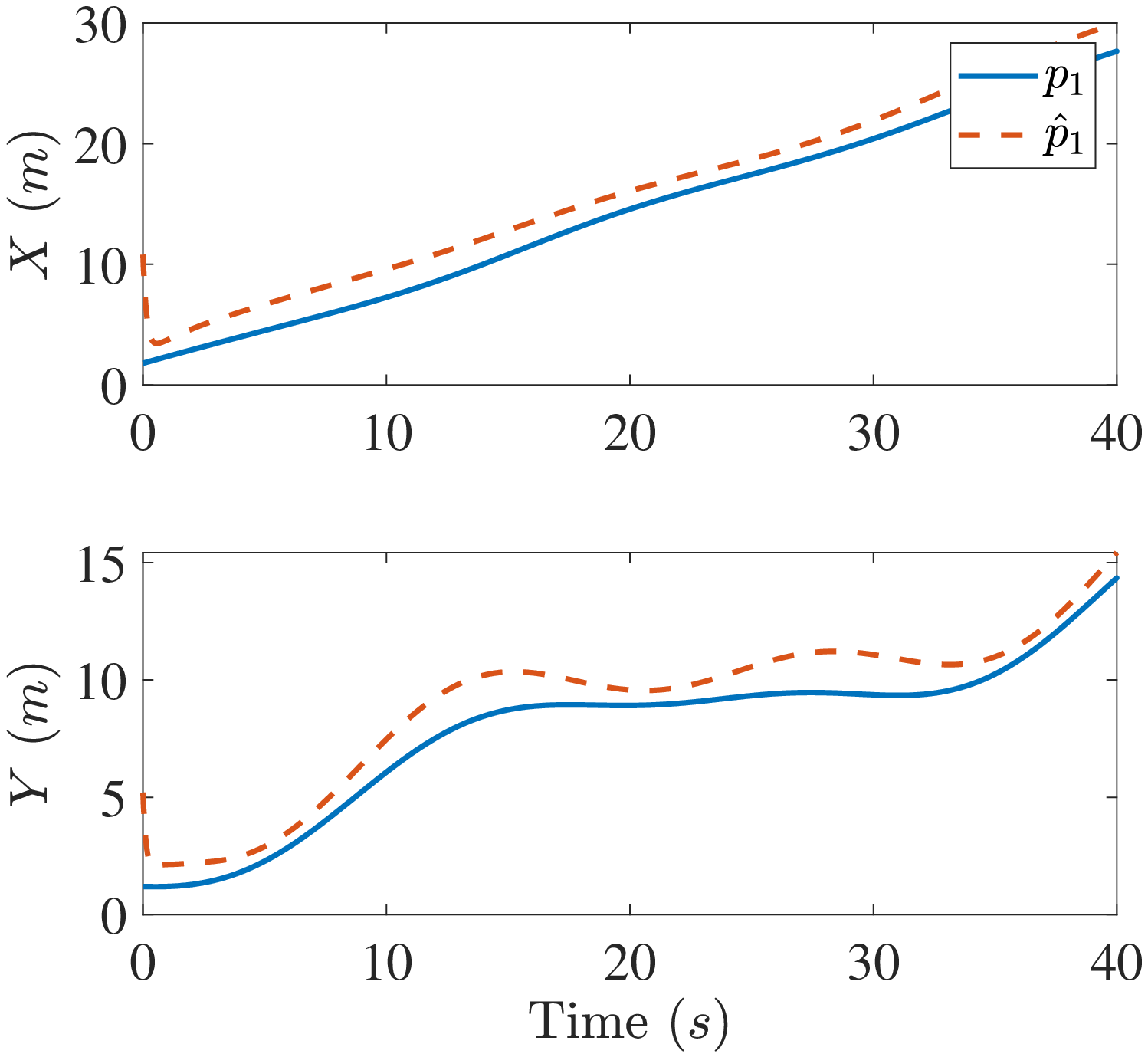}
		\caption{}
		\label{fig_consensus_p_privacy}
	\end{subfigure}
	\begin{subfigure}[b]{0.33\textwidth}
		\centering
		\includegraphics[width=2.2 in]{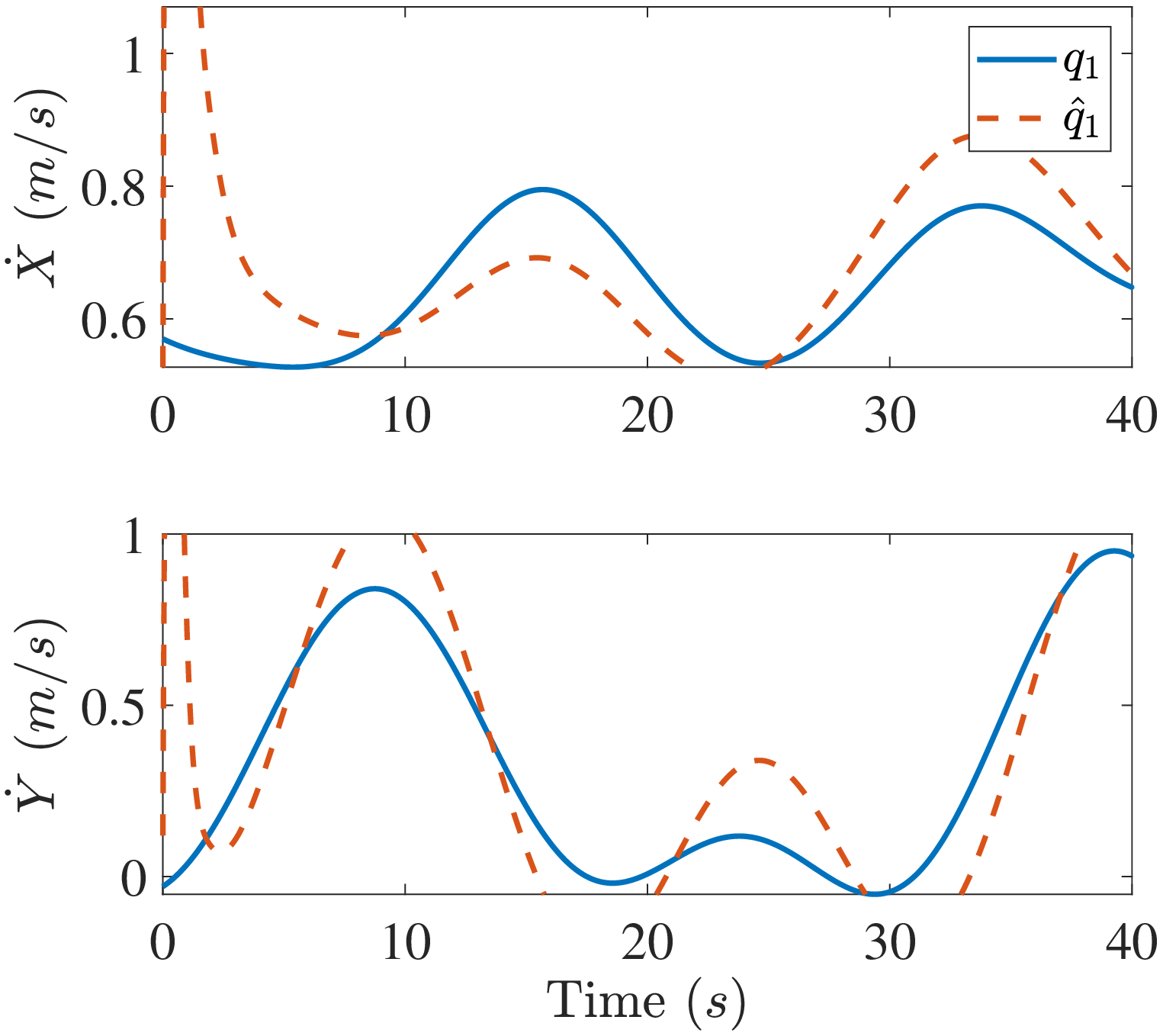}
		\caption{}
		\label{fig_consensus_q_privacy}
	\end{subfigure}
	\caption{Simulation results of privacy-preserving dynamic average consensus \eqref{sub_c}: (a) System state convergence. (b) Estimation of $p_{1}(t)$ with the eavesdropping scheme developed in Section \ref{sec_eavesdrop}. (c) Estimation of $q_{1}(t)$ with the eavesdropping scheme developed in Section \ref{sec_eavesdrop}.}
	\label{fig_sim2}
\end{figure*}

\begin{theorem} \label{theorem_formation}
	The controller designed in (\ref{c_input}) ensures that the system errors $e_{xi}(t)$, $e_{yi}(t)$, and $e_{\theta i}(t)$ ($i=1, 2, \cdots, n$) asymptotically converge to zero in the sense that
	\begin{equation} \label{lim_e}
	\lim_{t\rightarrow \infty} e_{xi}(t), e_{yi}(t), e_{\theta i}(t)=0.
	\end{equation}
\end{theorem} 

\begin{proof}
	See Appendix A.
\end{proof}

According to the definition of $e_{xi}(t)$, $e_{yi}(t)$, and $e_{\theta i}(t)$, it can be concluded that $s_{i}(t)\rightarrow c_{i}^{\alpha}(t) + b_{i}(t)$ as \eqref{lim_e} holds. Since $c_{i}^{\alpha}(t)$ will converge to the neighbourhood of the geometric center $\frac{1}{n} \sum_{i=1}^{n}p_{i}(t)$, the formation control task is accomplished as \eqref{lim_e} holds.

\section{Simulation Results}

In this section, simulation is conducted to demonstrate the performance of the developed approach. A team of four mobile robots are employed to follow a group of four mobile targets and maintain a rectangle formation. The network structure of the mobile robots is the same as the one shown in Figure~\ref{fig_formation}. The initial positions and velocities of mobile targets are as follows:
\[
\begin{aligned}
p_{1}(0) &= \begin{bmatrix}
1.8 \\
1.2
\end{bmatrix},
p_{2}(0) = \begin{bmatrix}
-1.2 \\
1.8
\end{bmatrix},
\\
p_{3}(0) &= \begin{bmatrix}
-1.8 \\
-1.2
\end{bmatrix},
p_{4}(0) = \begin{bmatrix}
1.2 \\
-1.8
\end{bmatrix},
\\
q_{1}(t) &= q_{0}(t) + \begin{bmatrix}
0.1\cos(0.2t) \\
-0.2\cos(0.4t)
\end{bmatrix},
\\
q_{2}(t) &= q_{0}(t) + \begin{bmatrix}
-0.2\cos(0.4t) \\
0.1\cos(0.2t)
\end{bmatrix},
\\
q_{3}(t) &= q_{0}(t) + \begin{bmatrix}
-0.1\cos(0.2t) \\
0.2\cos(0.4t)
\end{bmatrix},
\\
q_{4}(t) &= q_{0}(t) + \begin{bmatrix}
0.2\cos(0.4t) \\
-0.1\cos(0.2t)
\end{bmatrix},
\end{aligned}
\]
where $q_{0}(t) \in \mathbb{R}^{2}$ is given by
\[
q_{0}(t)= (0.75-0.25\cos(0.24t)) \begin{bmatrix}
\cos(\frac{\pi}{9}+0.5\sin(0.2t)) \\ \sin(\frac{\pi}{9}+0.5\sin(0.2t))
\end{bmatrix}.
\]
Furthermore, the initial positions of the mobile robots are selected as $s_{1}(0)=\begin{bmatrix}
1.3 & 5.2
\end{bmatrix}^{T}$, $s_{2}(0)=\begin{bmatrix}
-7.5 & 2.6
\end{bmatrix}^{T}$, $s_{3}(0)=\begin{bmatrix}
-4 & -5.5
\end{bmatrix}^{T}$, and $s_{4}(0)=\begin{bmatrix}
5.2 & -5.2
\end{bmatrix}^{T}$. 
%The desired formation pattern is a rectangle. 
For the mobile robots, the desired relative positions to the geometric center of mobile targets are given by $b_{1}(0)=\begin{bmatrix}
4 & 4
\end{bmatrix}^{T}$, $b_{2}(0)=\begin{bmatrix}
-4 & 4
\end{bmatrix}^{T}$, $b_{3}(0)=\begin{bmatrix}
-4 & -4
\end{bmatrix}^{T}$, and $b_{4}(0)=\begin{bmatrix}
4 & -4
\end{bmatrix}^{T}$.
In the following, we first evaluate the state decomposition based dynamic average consensus algorithm and then test the formation controller designed in \eqref{c_input}.

Suppose that an external eavesdropper is interested in obtaining the information of mobile target 1 and uses the eavesdropping scheme developed in Section \ref{sec_eavesdrop} to infer $p_{1}(t)$ and $q_{1}(t)$. To better demonstrate the performance of the proposed consensus scheme, both the conventional algorithm in \eqref{dot_x} and the developed privacy-preserving algorithm are used to estimate the geometric center of mobile targets. Figure~\ref{fig_sim1} shows the evolution of the network states as well as the eavesdropping states under the conventional algorithm \eqref{dot_x}. It can be seen that the eavesdropper can successfully infer $p_{1}(t)$ and $q_{1}(t)$ when the agents are updated with algorithm \eqref{dot_x}. The performance of the privacy-preserving scheme \eqref{sub_c} is illustrated in Figure~\ref{fig_sim2}. It is clear that the proposed scheme can achieve dynamic average consensus while protect the privacy of the mobile target.

As discussed in Section \ref{sec_formation}, the dynamic average consensus algorithm is used to estimate the geometric center of mobile targets, and then the mobile robots are driven with the controller \eqref{c_input} to achieve the formation task. Figure~\ref{fig_formation} depicts the motion trajectories of all mobile robots, showing that all robots follow the geometric center by spreading out in a desired rectangle pattern. Moreover, the evolution of the formation tracking errors are presented in Figure~\ref{fig_formation_error}. From Figure~\ref{fig_formation_error}, it can be seen that the tracking errors $e_{xi}(t)$, $e_{yi}(t)$, and $e_{\theta i}(t)$ all asymptotically converge to small values that are close to zero.
\begin{figure}[!htbp]
	\centering
	\hspace{0 in}
	\includegraphics[width=3.3 in]{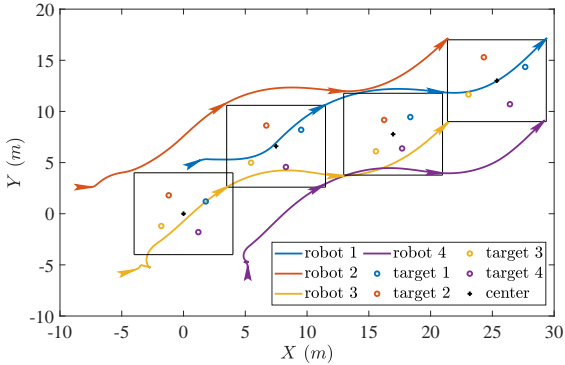}
	\caption{Motion trajectories of all mobile robots.}
	\label{fig_formation_motion}
\end{figure}
\begin{figure}[!htbp]
	\centering
	\hspace{0 in}
	\includegraphics[width=3.3 in]{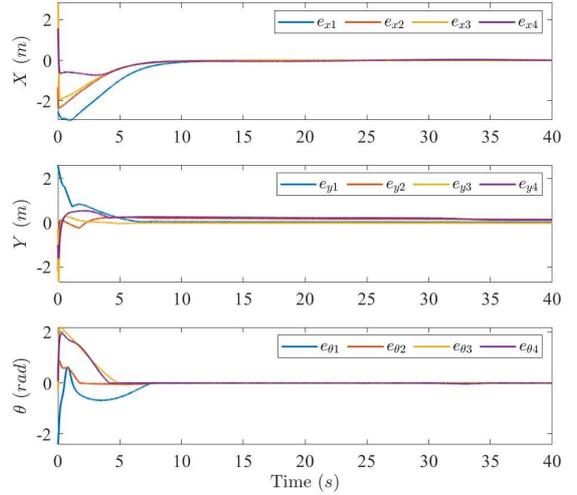}
	\caption{Formation control errors of each mobile robot.}
	\label{fig_formation_error}
\end{figure}

\section{Conclusion}
This paper developed a state decomposition based privacy-preserving method for continuous-time dynamic average consensus. 
%The key idea of the developed method was to decompose the original state into two sub-states. By suitably selecting the new reference signals of the two sub-states, the convergence properties of the consensus algorithm can be retained, and the reference signals of the original state cannot be discovered by the external eavesdropper. 
We showed that existing dynamic consensus algorithm is susceptible to eavesdropping attacks with a carefully designed filter. We then rigorously proved that the state decomposition scheme can enable privacy preservation without affecting the consensus results.
Furthermore, the proposed method was successfully applied to achieve formation control for non-holonomic mobile robots. Simulation results showed that by using the proposed method, the group of networked mobile robot can spread out in a pre-specified formation without disclosing private information.

\appendix
\section{Proof of Theorem \ref{theorem_formation}}
\begin{proof}
	To prove Theorem \ref{theorem_formation}, the Lyapunov function $V_{i}(t) \in \mathbb{R}, i=1, 2, \cdots, n$ is defined as
	\begin{equation} \label{Vi}
	V_{i}(t) \triangleq \frac{1}{2}\gamma_{4}\left(e_{xi}^{2}+e_{yi}^{2}\right) + \frac{1}{2}\bar{e}_{\theta i}^{2}.
	\end{equation}
	Based on \eqref{dot_si}, \eqref{ei}-\eqref{c_input} and the facts that $\dot{c}_{xi}(t) = v_{di}(t)\cos(\theta_{di}(t))$, $\dot{c}_{yi}(t) = v_{di}(t)\sin(\theta_{di}(t))$, the closed-loop error dynamics can be derived, as follows: 
	\begin{equation} \label{dot_ei}
	\begin{aligned}
	\dot{e}_{xi}(t) 
	%&= v_{i}(t)+\omega_{i}(t)e_{yi}(t) - \cos(e_{\theta i}(t))v_{di}(t)
	%\\
	&= -\gamma_{1}\tanh(e_{xi}(t))+\omega_{i}(t)e_{yi}(t),
	\\
	\dot{e}_{yi}(t) &= -\omega_{i}(t)e_{xi}(t) + \sin(e_{\theta i}(t))v_{di}(t),
	\\
	\dot{\bar{e}}_{\theta i}(t) &= -\gamma_{2}\tanh(\bar{e}_{\theta i}(t)) - \dot{\theta}_{di}(t) - \gamma_{3} \text{sgn}(\bar{e}_{\theta i}(t))
	\\
	& \quad \, - \gamma_{4}\frac{\sin(e_{\theta i}(t)) - \sin(\rho_{i}(t))}{\bar{e}_{\theta i}(t)} v_{di}(t) e_{yi}(t).
	\end{aligned}
	\end{equation}	
	After taking the time derivative of \eqref{Vi} and substituting \eqref{dot_ei} into the derivative, it can be determined that
	\begin{equation} \label{dot_Vi}
	\begin{aligned}
	\dot{V}_{i}(t) &= -\gamma_{1}\gamma_{4}e_{xi}(t)\tanh(e_{xi}(t))-\gamma_{2}\bar{e}_{\theta i}(t)\tanh(\bar{e}_{\theta i}(t)) 
	\\
	& \quad \, + \gamma_{4} v_{di}(t)e_{yi}(t)\sin(\rho_{i}(t)) 
	\\
	& \quad \, - \bar{e}_{\theta i}(t)\left( \gamma_{3} \text{sgn}(\bar{e}_{\theta i}(t)) + \dot{\theta}_{di}(t) \right).
	\end{aligned}
	\end{equation}
	If $\gamma_{3}$ is selected sufficiently large to satisfy $\gamma_{3} > \sup_{t\in \left[0, \infty\right)} |\dot{\theta}_{di}(t)|$, then $\dot{V}_{i}(t)$ is upper bounded by
	\begin{equation} \label{Vi_bound}
	\begin{aligned}
	\dot{V}_{i}(t) &\le -W_{i}(t) + \gamma_{4} |e_{yi}(t)| |v_{di}(t)\sin(\rho_{i}(t))|
	\\
	& \le  \sqrt{2 \gamma_{4} V_{i}(t)} |v_{di}(t)\sin(\rho_{i}(t))|,
	\end{aligned}
	\end{equation}
	where $W_{i}(t) \in \mathbb{R}$ is a non-negative function given by
	\begin{equation} \label{Wi}
	W_{i} \triangleq \gamma_{1}\gamma_{4}e_{xi}(t)\tanh(e_{xi}(t)) + \gamma_{2}\bar{e}_{\theta i}(t)\tanh(\bar{e}_{\theta i}(t)).
	\end{equation}
	From \eqref{rho_i}, it can be found that $0 \le \varpi_{i}(t) \le 1$, $\dot{\varpi}_{i}(t) = -|v_{di}(t)|\varpi_{i}(t)$, and $|\rho_{i}(t)| \le \iota_{0} \varpi_{i}(t)$. Using these facts and integrating $|v_{di}(t) \sin(\rho_{i}(t)) |$, 
	it can be concluded that
	\begin{equation} \label{L1}
	\begin{aligned}
	\int_{0}^{t} |v_{di}(\tau) \sin(\rho_{i}(\tau)) |d\tau 
	\le & \int_{0}^{t} |v_{di}(\tau)| |\rho_{i}(\tau) |d\tau
	\\
	\le & \iota_{0} \int_{0}^{t} |v_{di}(\tau)|\varpi_{i}(\tau) d\tau
	\\
	\le & \iota_{0} \int_{0}^{t} -\dot{\varpi}_{i}(\tau) d\tau
	\\
	\le &\iota_{0} \left(\varpi(0)-\varpi(t)\right) \le \iota_{0}.
	\end{aligned}
	\end{equation}
	Eqn. \eqref{Vi_bound} indicates that $\frac{d\sqrt{V_{i}(t)}}{dt} \le \sqrt{\frac{\gamma_{4}}{2}} |v_{di}(t)\sin(\rho_{i}(t))|$, and then based on \eqref{L1}, it can be deduced that $\sqrt{{V}_{i}(t)} \in \mathcal{L}_{\infty}$, i.e., $V_{i}(t) \in \mathcal{L}_{\infty}$. Furthermore, it can be inferred from \eqref{c_input}, \eqref{Vi}, and \eqref{dot_ei} that $e_{xi}(t)$, $e_{yi}(t)$, $\bar{e}_{\theta i}(t)$, $v_{i}(t)$, $\omega_{i}(t)$, $\dot{e}_{xi}(t)$, $\dot{e}_{yi}(t)$, $\dot{\bar{e}}_{\theta i}(t) \in \mathcal{L}_{\infty}$. 
	Taking the time derivative of $W_{i}(t)$ and using the above boundedness analysis, it can be derived that $\dot{W}_{i}(t) \in \mathcal{L}_{\infty}$, which is a sufficient condition for $W_{i}(t)$ being uniformly continuous. 
	Using \eqref{Vi_bound}, \eqref{L1}, and $V_{i}(t) \in \mathcal{L}_{\infty}$, it can be concluded that $\int_{0}^{t} W_{i}(\tau) d\tau \in \mathcal{L}_{\infty}$. Based on $\int_{0}^{t} W_{i}(\tau) d\tau \in \mathcal{L}_{\infty}$ and the uniform continuity of $W_{i}(t)$, Barbalat's lemma \cite{khalil2002nonlinear} can be exploited to obtain that $\lim_{t \rightarrow \infty} W_{i}(t) = 0$, i.e., $\lim_{t \rightarrow \infty} e_{xi}(t), \bar{e}_{\theta i}(t) = 0$.
	With the aid of the extended Barbalat's lemma \cite{Dixon2000Nonlinear,Samson1995Control}, it can be further deduced that $\lim_{t \rightarrow \infty} e_{yi}(t) = 0$. According to \eqref{bar_ei} and \eqref{rho_i}, it is clear that $\lim_{t \rightarrow \infty} e_{xi}(t), e_{yi}(t), \bar{e}_{\theta i}(t) = 0$ implies $\lim_{t \rightarrow \infty} e_{xi}(t), e_{yi}(t), e_{\theta i}(t) = 0$, which completes the proof of Theorem \ref{theorem_formation}.
\end{proof}	

\bibliographystyle{plain}
\bibliography{IEEEfull,reference}

\end{document}